\documentclass{article}
\usepackage[round]{natbib}

\usepackage{algorithm}
\usepackage{algorithmic}

\usepackage{amsfonts,amsmath,amssymb,mathtools,xfrac}
\usepackage{xspace}
\usepackage[margin=1in]{geometry}
\usepackage{natbib}
\usepackage{hyperref}
\usepackage{tikz}
\hypersetup{%
   breaklinks,%
   colorlinks=true,%
   linkcolor=black,%
   urlcolor=black,%
   citecolor=[rgb]{0,0,0.45}
}
\usepackage{xcolor}
\usepackage{enumitem}
\usepackage{booktabs}

\usepackage[amsmath,thmmarks]{ntheorem}%
\theoremstyle{plain}%
\newtheorem{theorem}{Theorem}[section]

\newtheorem{result}{Result}
\newtheorem{lemma}[theorem]{Lemma}

\newtheorem{corollary}[theorem]{Corollary}
\newtheorem{claim}[theorem]{Claim} 

\newtheorem{observation}[theorem]{Observation}

\newtheorem{definition}[theorem]{Definition}

\theoremstyle{plain}%
\theoremheaderfont{\bf} \theorembodyfont{\upshape}%
\newtheorem{remark}[theorem]{Remark}%
\newtheorem*{remark:unnumbered}[theorem]{Remark}%
\theoremheaderfont{\em}%
\theorembodyfont{\upshape}%
\theoremstyle{nonumberplain}%
\theoremseparator{}%
\theoremsymbol{\myqedsymbol}%
\newtheorem{proof}{Proof:}%

\newcommand{\myqedsymbol}{$\square$}

\newtheorem{proofof}{Proof of\!}%

\newcommand{\eps}{\varepsilon}%
\def\bar{\overline}

\def\script#1{\mathcal{#1}}

\def\sep{\;|\;}

\newcommand{\argmin}{\mathrm{argmin}{}\xspace}

\DeclareMathSymbol{\shortminus}{\mathbin}{AMSa}{"39}

\def\sB{\script{B}}
\def\sC{\script{C}}
\def\sD{\script{D}}

\def\sF{\script{F}}

\def\sI{\script{I}}

\def\sM{\script{M}}

\def\sP{\script{P}}
\def\sQ{\script{Q}}
\def\sR{\script{R}}
\def\sS{\script{S}}

\def\cllp{\textsc{FacilityMatLP}}
\def\clla{\textsc{FacilityMatAlg}}
\def\cllc{\textsc{CenterPartAlg}}

\def\opt{\textsc{OPT}}

\def\kcenter{k\mathrm{center}}

\newcommand\cost{{\operatorname{cost}}}

\def\R{\script{R}}
\newcommand{\sol}{\textsc{SOL}}%

\def\cst{T}
\def\pcst{H}
\def\solsp{\mathcal{P}}
\def\intSolSp{\mathcal{Q}}

\def\flt{\sF''}

\def\core{core}
\def\pri{p}
\def\sec{s}

\title{Improved Approximation Algorithms for Individually Fair Clustering}

\author{Ali Vakilian\thanks{Toyota Technological Institute at Chicago (TTIC). Supported by NSF award CCF-1934843. Email: \texttt{vakilian@ttic.edu}} \and Mustafa Yal\c{c}{\i}ner\thanks{Technical University of Dortmund. Email: \texttt{mustafa.yalciner@tu-dortmund.de}}}

\begin{document}
\date{}
\maketitle
\begin{abstract}
We consider the $k$-clustering problem with $\ell_p$-norm cost, which includes $k$-median, $k$-means and $k$-center, under an individual notion of fairness proposed by~\cite{jung2019center}: given a set of points $P$ of size $n$, a set of $k$ centers induces a {\em fair clustering} if every point in $P$ has a center among its $n/k$ closest neighbors. \cite{mahabadi2020individual} presented a $( p^{O(p)},7)$-bicriteria approximation for fair clustering with $\ell_p$-norm cost: every point finds a center within distance at most $7$ times its distance to its $(n/k)$-th closest neighbor and the $\ell_p$-norm cost of the solution is at most $p^{O(p)}$ times the cost of an optimal fair solution.

In this work, for any $\eps>0$, we present an improved $(16^p +\eps,3)$-bicriteria for this problem. Moreover, for $p=1$ ($k$-median) and $p=\infty$ ($k$-center), we present improved cost-approximation factors $7.081+\eps$ and $3+\eps$ respectively. 
To achieve our guarantees, we extend the framework of~\citep{charikar2002constant,swamy2016improved} and devise a $16^p$-approximation algorithm for the facility location with $\ell_p$-norm cost under matroid constraint which might be of an independent interest. 

Besides, our approach suggests a reduction from our individually fair clustering to a clustering with {\em a group fairness requirement} proposed by~\cite{kleindessner2019fair}, which is essentially the {\em median matroid} problem~\citep{krishnaswamy2011matroid}.    
\end{abstract}

\section{Introduction}
As automated decision-making is widely used in a diverse set of important decisions such as job hiring, loan application approval and college admission, there is a debate regarding the fairness of algorithms and machine learning methods. As there are lots of instances in which optimizing machine learning algorithms with respect to the classical measures of efficiency (e.g., accuracy, runtime and space complexity) lead to biased outputs, e.g.,~\citep{Imana21a,angwin2016machine}, there is an extensive literature on {\em algorithmic fairness} which includes both {\em how to define the notion of fairness} and {\em how to design efficient algorithms with respect to fairness constraints}~\citep{dwork2012fairness,chouldechova2017fair,chouldechova2018frontiers,kearns2019ethical}.
Clustering is one of the basic tasks in unsupervised learning and is a commonly used technique in many fields such as pattern recognition, information retrieval and data compression. Due to its wide range of applications, the clustering problem has been studied extensively under fairness consideration. Fair clustering was first introduced in a seminal work of~\cite{chierichetti2017fair} where they proposed the {\em balanced clusters} as the notion of fairness. Further, other {\em group fairness} notions such as {\em balanced centers}~\citep{kleindessner2019fair} and {\em balanced costs}~\citep{abbasi2020fair,ghadiri2020fair} were also introduced as measures of fairness.  

While clustering under group fairness is a well-studied domain by now, we know much less about the complexity of fairness under {\em individual fairness}. Motivated by the interpretation of clustering as a {\em facility location} problem,~\cite{jung2019center} proposed an individual notion of fairness for clustering as follows: a clustering of a given pointset $P$ is {\em fair} if every point in $P$ has a center among its $(|P|/k)$-closest neighbors. To justify, if a set of $k$ centers are supposed to be opened, then, without any prior knowledge, each point (or client) will expect to find a center among $1/k$ fraction of points that are closest to it. This is in particular a reasonable expectation in many scenarios. For example, people living in areas with different densities have different expectations for a ``reasonably close distance''. So, while in an urban area of a major city it is reasonable for a resident to find a grocery store within a mile of her home, it is a less reasonable expectation for someone who lives in a low-density rural area.~\cite{jung2019center} proves that it is NP-hard to find a fair clustering and proposed an algorithm that returns a $2$-approximate fair clustering in any metric space---each point has a center at distance at most twice the distance to its $(|P|/k)$-closest neighbor.\footnote{Here, we assume that each point is the (first) closest neighbor to itself.}   

Recently,~\cite{mahabadi2020individual} extended this notion of fairness and studied the common center-based clustering objective functions such as $k$-median, $k$-means and $k$-center under this individual notion of fairness. 
In particular,~\citep{mahabadi2020individual} showed that a local search type algorithm achieves a bicriteria approximation guarantee for the aforementioned clustering problem with individual fairness. More generally, they considered the $\alpha$-fair $k$-clustering with $\ell_p$-norm cost, $\min_{C\subseteq P} \sum_{v\in P} d(v,C)^p$, and proved that a local search algorithm with swaps of size at most $4$ finds a $(p^{O(p)},7)$-bicriteria approximate solution: every point has a center at distance at most $7$ times its ``desired distance'' and the $\ell_p$-clustering cost of the solution is at most $p^{O(p)}$ times the optimal fair $k$-clustering (refer to Section~\ref{sec:reduction} for more details). Given a pointset of size $n$ and a fairness parameter $\alpha\geq 1$, for every point $v$, the {\em desired distance} of $v$ is $\alpha$ times its {\em fair radius} where the fair radius is the distance of the $(n/k)$-th closest neighbor of $v$ in the pointset.

\subsection{Our Contributions}
In this paper, we study the problem of $\alpha$-fair $k$-clustering with $\ell_p$-norm cost function and improve upon both fairness and cost approximation factors of the $(p^{O(p)},7)$-bicriteria approximation of~\citep{mahabadi2020individual} significantly. 
In particular, our result improves upon the $(O(\log n),7)$-bicriteria approximation of fair $k$-center clustering and achieves a $(O(1),3)$-bicriteria approximation.

\begin{result}[restatement of Theorem~\ref{thm:main}]\label{result:lp-fair-clustering}
For any $\eps>0$, $\alpha\geq 1$ and $p>1$, there exists a $(16^p+\eps,3)$-bicriteria approximation algorithm for $\alpha$-fair $k$-clustering with $\ell_p$-norm cost. Moreover, for $p=1$, which denotes the $\alpha$-fair $k$-median problem, there exists a $(7.081+\epsilon,3)$-bicriteria approximation algorithm.
\end{result}

We remark that for fair $k$-median, our result improves upon the $(84,7)$-bicriteria approximation of~\cite{mahabadi2020individual}.

To achieve our approximation guarantees, we design an $e^{O(p)}$-approximation for the problem of {\em facility location with $\ell_p$-norm cost under matroid constraint}. This is a natural generalization of the well-known facility location problem under matroid constraint~\citep{krishnaswamy2011matroid,swamy2016improved} which includes the {\em matroid median} problem as its special case. Our approach extends the algorithm of~\cite{swamy2016improved} where we show that a careful modification of the analysis obtains the desired approximation guarantee for the more general problem of facility location with $\ell_p$-norm cost.\footnote{We remark that one can apply the same approach and reduce the problem to an instance of $k$-clustering under matroid constraint instead. This still requires a generalization of matroid-median problem with $\ell_p$-cost.} We remark that for the case of $k$-median ($p=1$), we can instead employ the best-known bound for matroid $k$-median  by~\cite{krishnaswamy2018constant} and get $(7.081+\epsilon)$-approximation.

\begin{result}[restatement of Theorem~\ref{thm:main-facility}]\label{result:lp-facility-matroid}
For any $p\in [1, \infty)$, there exists a $16^p$-approximation algorithm for the facility location problem with $\ell_p$-norm cost under matroid constraint.
\end{result}



Besides our theoretical contributions, our approach essentially draws an interesting connection between the individual fairness and the group fairness notion {\em with balanced centers}. 
In particular, we show that a ``density-based'' decomposition of the points introduces a set of groups such that a balanced representation of them in the centers guarantees a fair solution w.r.t.~the individual fairness notions considered in this paper. This observation could be of an independent interest as to the best of our knowledge is the first to connect two different notions of fairness that have been introduced for the clustering problem.           
Besides, the connection between our notion of individual fairness and the notion of group fairness with balanced centers has led to an improved algorithm for the fair $k$-center problem---we show a $(3,3)$-bicriteria approximation for $\alpha$-fair $k$-center problem. Unlike our main approach, this result only holds for the $k$-center problem and crucially relies on properties of $k$-center objective function and a recent $3$-approximation algorithm for $k$-center with balanced center~\citep{jones2020fair}.

\begin{result}[restatement of Theorem~\ref{thm:fair-k-center}]\label{result:fair-k-center}
For any $\eps>0$ and $\alpha\geq 1$, there exists a $(3, 3)$-bicriteria approximation algorithm for $\alpha$-fair $k$-center.
\end{result}


\subsection{Other Related Work}
\paragraph{Clustering with group fairness constraint.}
\cite{chierichetti2017fair} introduced the first notion of fair clustering with balanced clusters: given a set of points coming from two distinct groups, the goal is to find a minimum {\em cost} clustering with proportionally balanced clusters. 
Their approach, which is based on a technique called {\em fairlet decomposition}, achieves constant factor approximations for fair $k$-center and $k$-median. Since then, several variants of clustering w.r.t.~a notion of group fairness have been studied.
\begin{itemize}[leftmargin=*]
    \item{\bf With balanced clusters.} This is the first and the most well-studied notion of fair clustering. In a series of work, this setting has been extended to address general $\ell_p$-norm cost function, multiple groups, relaxed balanced requirements (with both upper and lower bound on ratio of each class in any cluster) and scalability issues~\citep{chierichetti2017fair,backurs2019scalable,bera2019fair,bercea2019cost,ahmadian2019clustering,schmidt2019fair,huang2019coresets}.  
    
    \item{\bf With balanced centers.} Another notion of group fairness, proposed by~\cite{kleindessner2019fair}, aims to minimize the $k$-center cost function and guarantee a fair representation of groups in the selected centers. Their notion is essentially $k$-center under partition matroid constraint. As mentioned earlier in the paper, our approach studies a generalization of this problem, facility location with $\ell_p$-norm cost under matroid constraint, as a subroutine. 
    Recently,~\cite{jones2020fair} designed a $3$-approximation algorithm for the fair $k$-center with balanced centers that runs in time $O(nk)$. We remark that other clustering objective functions, in particular $k$-median, have been studied extensively under the partition matroid constraint, and more generally matroid constraint too~\citep{hajiaghayi2010budgeted,krishnaswamy2011matroid,charikar2012dependent,chen2016matroid,swamy2016improved,krishnaswamy2018constant}. A similar notion has been studied for the related {\em nearest neighbor} problem~\citep{har2019near,aumuller2020fair,aumuller2021sampling}.
    \item{\bf With balanced cost.} Recently,~\cite{abbasi2020fair,ghadiri2020fair} independently proposed a notion of fair clustering, called {\em socially fair} clustering, in which the goal is to minimize the maximum cost that any group in the input pointset incurs.~\cite{makarychev2021approximation} designed an algorithm that improves upon the $O(\ell)$-approximation of~\citep{abbasi2020fair,ghadiri2020fair} for socially fair $k$-means and $k$-median and achieves $O(\log\ell / \log\log \ell)$-approximation where $\ell$ denotes the number of different groups in the input. The objective of socially fair clustering was previously studied in the context of {\em robust} clustering~\citep{anthony2010plant}. In this notion of robust algorithms, a set $\sS$ of possible input scenarios are provided in the input and the goal is to output a solution which is simultaneously ``good'' for all scenarios.~\cite{anthony2010plant} gave an $O(\log n + \log \ell)$-approximation for robust $k$-median and a set of related problems in this model on an $n$-point metric space. Moreover,~\cite{bhattacharya2014new} showed that it hard to approximate robust $k$-median by a factor better than $\Omega(\log \ell / \log \log \ell)$ unless $\mathrm{NP} \subseteq \bigcap_{\delta>0} \mathrm{DTIME}(2^{n^\delta})$ which essentially shows that the approximation guarantee of~\cite{makarychev2021approximation} for socially fair $k$-median is tight up to a constant factor. Very recently,~\cite{chlamtavc2022approximating} studied a more general notion of $(p,q)$-fair clustering which captures socially fair clustering as a special case and its objective smoothly interpolates between the objectives of $k$-clustering with $\ell_p$-cost and socially fair clustering with $\ell_p$-cost.
\end{itemize}
Inspired by the recent work on the fair allocation of public resources,~\cite{chen2019proportionally} introduced a notion of fair $k$-clustering as follows. Given a set of $n$ points in a metric space, a set of $k$ centers $C$ is fair if no subset of $n/k$ points $M\subset P$ have the incentive to assign themselves to a center outside $C$; there is no point $c'$ outside $C$ such that the distance of all points in $M$ to $c'$ is smaller than their distance to $C$.~\cite{chen2019proportionally,micha2020proportionally} devised approximation algorithms for several variants of this notion of fair clustering.   
\paragraph{Clustering with individual fairness constraint.}~\cite{kleindessner2020notion} studied a different individual notion of fairness in which a point is treated fairly, if its cluster is ``stable''---the average distance of the point to its own cluster is {\em not larger} than the average distance of the point to the points of any other cluster. They proved that in a  general metric, even deciding whether such a fair $2$-clustering exists is NP-hard. Further, they showed that such fair clustering exists in one dimensional space for any values of $k$.  

\cite{anderson2020distributional} proposed a distributional individual fairness for $\ell_p$-norm clustering where each point has to be mapped to a selected set of centers according to a probability distribution over the centers and then the goal is to minimize the expected $\ell_p$-norm clustering cost while ensuring that ``similar'' points have ``similar'' distributional assignments to the centers.  

\paragraph{Connections to priority $k$-center.} The proposed notion of individual fairness by~\cite{jung2019center} which we consider in this paper was also studied in different contexts such as {\em priority clustering} (or clustering with usage weights)~\citep{plesnik1987heuristic} and metric embedding~\citep{chan2006spanners,charikar2010local}. We remark that all of theses results imply a $2$-approximation algorithm for fair clustering. However, as in the work of~\cite{jung2019center}, all these result only find an approximately fair clustering and does not minimize any global clustering cost functions such as $k$-center, $k$-median and $k$-means.

Parallel and independent to this work,~\cite{chakrabarty2021better} presented an $(8, 2^{1 + 2/p})$-biceriteria approximation algorithm for the individually fair $k$-clustering with $\ell_p$-norm cost problem. In particular, their approach achieves $(8,8)$, $(8,4)$ and $(8, 2+\eps)$ (for arbitrarily small values of $\eps$) for fair $k$-median, fair $k$-means and fair $k$-center. Table \ref{fig:table:results} provides an overview of the existing results.

\begin{table*}[t]
\centering
    \begin{tabular}{lrrrrrr}\toprule
        &\multicolumn{2}{c}{\textbf{$k$-median}}&\multicolumn{2}{c}{\textbf{$k$-means}}&\multicolumn{2}{c}{\textbf{$k$-center}}
        \\\cmidrule(r){2-3}\cmidrule(r){4-5}\cmidrule(r){6-7}   
        &Cost&Fairness&Cost&Fairness&Cost&Fairness\\\midrule
        \cite{mahabadi2020individual}    & $84$ & $7$
                & $O(1)$ & $7$
                & $O(\log(n))$ & $7$\\
        \cite{chakrabarty2021better}   & $8$ & $8$
                & $4$ & $8$
                &  $2+\eps$ & $8$ \\
        Ours     & $7.081+\epsilon$ & $3$ 
                & $16+\eps$ & $3$
                & $3+\eps$ & $3$
        \\\bottomrule
    \end{tabular}
    \caption{Comparison of the results, where $\eps>0$ is an arbitrarily small variable.}\label{fig:table:results}
\end{table*} 

We remark that \cite{chakrabarty2021better} implemented their algorithm and used a parameterized sparsification technique to configure the trade-off between the fairness/cost objective and computational complexity. 

\paragraph{Better bounds for fair $k$-median objective.} As mentioned, we can employ the improved result of~\cite{krishnaswamy2018constant} and obtain $(7.081+\eps, 3)$-approximation for fair $k$-median which strictly improves the recent bounds of~\citep{chakrabarty2021better}. Another improved bound related to matroid $k$-median is the recent result of~\citet{gupta2021structural}. However, we cannot apply~\citep{gupta2021structural} (in a black-box manner) and get a better approximation factor. For our application (i.e., in our partition matroids, we may have $\Theta(k)$ parts) their algorithms only guarantee a \textit{pseudo-approximation} guarantee---which assign fractional values to $O(k)$ facilities/centers---with approximation ratio $6.387 + \eps$. So, while via some pre- and post-processing they can derive a true $6.387$-approximation for $k$-median with knapsack constraint and $k$-median with outliers from their pseudo-approximation, their approach does not imply such approximation algorithms for the general setting with $\Theta(k)$ knapsack constraints which we need for the case of $k$-median with the partition matroid constraint.
\section{Preliminaries}\label{sec:prelim}
\begin{definition}[approximate triangle inequality]
A distance function $d$ satisfies the $\alpha$-approximate triangle inequality over a set of points $P$, if $\forall u,v,w \in P, d(u,w)\leq \alpha \cdot (d(u,v)+d(v,w))$
\end{definition}

\begin{observation}\label{obser:gen-triangle-ineq}
Let $(P,d)$ be a metric space. Then,
\begin{enumerate}[leftmargin=*]
    \item\label{item:lambda} For any $\lambda>0, p\ge 1$, The distance function $d(\cdot,\cdot)^p$ satisfies 
    \begin{align}\label{ineq:gen-triangle-ineq}
        d(u,v)^p 
        &\leq (1+\lambda)^{p-1} d(u,w)^p \nonumber\\
        &+ \Big(\frac{(1+\lambda)}{\lambda} \Big)^{p-1} d(w,v)^p.
    \end{align}
    In particular, for $p\ge 1$, the function $d(\cdot,\cdot)^p$ satisfies the $\alpha_p$-approximate triangle inequality for $\alpha_p=2^{p-1}$.
    \item For any $\lambda>0, p\ge 1$, The distance function $d(\cdot,\cdot)^p$ satisfies
    \begin{align}\label{ineq:2hop-triangle-ineq}
        d(u,v)^p \leq 3^{p-1} \cdot (d(u, w)^p + d(w, z)^p + d(z,v)^p).
    \end{align}
\end{enumerate}
\end{observation}
\begin{proof}
Note that Eq.~\eqref{ineq:gen-triangle-ineq} holds using Lemma~\ref{lem:p-norm-ineq} and the fact that $d(u,v) \leq d(u,w) + d(w,v)$. Furthermore, by setting $\lambda =1$, $d(\cdot, \cdot)^p$ satisfies the $\alpha_p$-approximate triangle inequality for $\alpha_p =2^{p-1}$.

Next, to prove the second inequality, Eq.~\eqref{ineq:2hop-triangle-ineq}, note that $d(u,v) \leq d(u, w) + d(w,z) + d(z,v)$. Then, by an application of Lemma~\ref{lem:p-norm-ineq} with $\lambda = 2$, we get Eq.~\eqref{ineq:2hop-triangle-ineq}. 
\end{proof}
\section{A Reduction from Fair Clustering to Facility Location Under Matroid Constraint}\label{sec:reduction}
In this section, we provide a reduction from our fair clustering problem to the problem of facility location under matroid constraint.
We use $P$ to denote the set of points in the input. We use $C\subseteq P$ to denote the subset of points that serve as centers. 
Throughout the paper, we consider the general $\ell_p$-norm cost function which is defined as bellow:  
\begin{equation}\label{eqn:ClusteringCost}
\cost(P, C; p) := \sum_{v\in P}d(v,C)^p,
\end{equation}
where $d(v,C)$ denotes the distance of $v$ to its closest center in $C$, i.e. $d(v,C) := \min_{c\in C}d(v,c)$. This cost function generalizes the cost functions corresponding to $k$-median ($p=1$), $k$-means ($p=2$) and $k$-center ($p=\infty$).\footnote{Note that for all $x\in \mathbb{R}^{n}$, $\left\|x\right\|_\infty \leq \left\|x\right\|_{\log n} \leq 2\left\|x\right\|_\infty$. This implies that setting $p=\log n$, the objective function $2$-approximates the objective of $k$-center.}.

Next, we set up some notations to formally define the notion of fairness we consider in this paper. For every point $v\in P$, we use $B(v,r) := \{u\in P: d(v,u) \leq r\}$ to denote the subset of all points in $P$ that are at distance at most $r$ from $v$ and call it the {\em ball} around $v$ with radius $r$.

\begin{definition}[fair radius]\label{def:FairRadius}
Let $P$ be a set of points of size $n$ in a metric space $(X,d)$ and let $\ell\in [n]$ be a parameter. For every point $v\in P$, we define the fair radius $r_{\ell}(v)$ to be the minimum distance $r$ such that $|B(v,r)| \geq n/\ell$. When $\ell=k$, we drop the subscript and use $r(\cdot)$ to denote $r_k(\cdot)$. 
\end{definition}
Here, we consider a more general variant of the problem studied by~\cite{mahabadi2020individual} as follows. 
\begin{definition}[$\alpha$-fair $k$-clustering]\label{def:alpha-fair-k-clustering}
Let $P$ be a set of points of size $n$ in a metric space $(X,d)$. A set of $k$ centers $C$ is $\alpha$-fair, if for every point $x\in P$, $d(x, C) \leq \alpha r_k(x)$. 
\end{definition}
Note that since even deciding whether a given set of points $P$ has a fair clustering or not is NP-hard~\cite{jung2019center} (i.e., $\alpha=1$), the best we can hope for is a {\em bicriteria approximation} guarantee. 

\begin{definition}[bicriteria approximation]\label{def:bicrit-approx}
An algorithm is a {\em $(\beta, \gamma)$-bicriteria approximation} for $\alpha$-fair $k$-clustering w.r.t.~a given $\ell_p$-norm cost function if for any set of points $P$ in the metric space $(X,d)$ the solution $\sol$ returned by the algorithm on $P$ satisfies the following properties:
\begin{enumerate}[leftmargin=*]
    \item $\cost(P,\sol; p) \leq \beta \cdot \cost(P,\opt; p)$ where $\opt$ denotes the optimal set of $k$ centers for $\alpha$-fair $k$-clustering of $P$ w.r.t.~the given $\ell_p$-norm cost function. In particular, $\cost(P,\opt;p)=\infty$ if an $\alpha$-fair $k$-clustering does not exist for $P$. \label{def:bicrit-approx-cost}
    \item $\sol$ is a $(\gamma\cdot \alpha)$-fair $k$-clustering of $P$. \label{def:bicrit-approx-fair}
\end{enumerate}
\end{definition}


Our main technical contribution is the following.
\begin{theorem}\label{thm:main}
For any $\eps>0$, $\alpha\geq 1$ and $p>1$, there exists a polynomial time algorithm that computes a $(16^p+\eps, 3)$-bicriteria approximate solution for $\alpha$-fair $k$-clustering with $\ell_p$-norm cost. Moreover, for $p=1$, which denotes the $\alpha$-fair $k$-median problem, there exists a $(7.081+\eps,3)$-bicriteria approximation algorithm.
\end{theorem}

The rest of the paper is to show the above theorem. To satisfy the fairness constraint, our approach relies on the existence of a special set of regions, called \textit{critical regions}.

\begin{definition}[critical regions]\label{def:CritBalls}
Let $P$ be a set of points in a metric space $(X,d)$ and let $\alpha$ be the desired fairness approximation.
A set of balls $\sB=\{B(c_1,\alpha \cdot r(c_1)),\ldots,B(c_m, \alpha \cdot r(c_m))\}$ where $m\leq k$
is called {\em critical regions}, if they satisfy the following properties:
\begin{enumerate}[leftmargin=*]
    \item For every $x\in P: d(x,\{c_1,\ldots, c_m\})\leq 2\alpha \cdot r(x)$ \label{def:critBalls:prop1}
    \item For any pair of centers $c_i,c_j$, $d(c_i,c_j) > 2\alpha \cdot \max\{r(c_i),r(c_j)\}$; in other words, critical regions are disjoint.
 \label{def:critBalls:prop2}
\end{enumerate}
\end{definition}

We now provide an algorithm that given a set of points $P$ and a fairness parameter $\alpha$, returns a set of critical regions. Our approach is similar to the methods proposed by~\cite{mahabadi2020individual} which is a slight modification of the greedy approach of~\citep{chan2006spanners,charikar2010local}. 
\begin{algorithm}[h]
	\begin{algorithmic}[1]
		\STATE {\bfseries Input:} Fairness parameter $\alpha$
		\STATE \textbf{initialize} covered points $Z\gets \emptyset$, centers of the selected balls $\sC \gets \emptyset$
		\WHILE{$Z \neq P$}
		\STATE $c \gets \argmin_{x \in P\setminus Z} r(x)$ \label{alg:AlphaCluster:NonDecreasingR}
		\STATE $\sC \gets \sC \cup \{c\}$\label{alg:AlphaCluster:ExpandCritArea}
		\STATE $Z \gets Z \cup \{x\in P\setminus Z| d(x,c) \leq 2\alpha\cdot r(x)\}$ \label{alg:AlphaCluster:CoverPoint}
		\ENDWHILE
		\RETURN $\{B(c,\alpha r(c)):c\in \sC\}$
    \end{algorithmic}
	\caption{outputs a set of critical regions for given parameters $\alpha, k$.}
	\label{alg:AlphaCluster}
\end{algorithm}


\begin{lemma}\label{lem:critical-region}
Let $P$ be a set of points of size $n$ in a metric space $(X,d)$, let $k$ be a positive integer and let $\alpha$ be a parameter denoting the desired fairness guarantee. Then, Algorithm~\ref{alg:AlphaCluster} returns a set of at most $k$ critical regions.
\end{lemma}
\begin{proof}
First we show that the set of centers returned by the algorithm satisfies property~\eqref{def:critBalls:prop1} of the critical regions. For every points $x\in P$ let $c_x$ denote the first center added to $\sC$ such that $x\in B(c_x, \alpha\cdot r(c_x))$. 
Hence, $d(x, \sC) \leq d(x, c_x) \leq  2\alpha\cdot r(x)$ where the last inequality follows from the fact that $c_x$ marks $x$ as covered.

Next, consider the iteration of the algorithm in which a center $c$ is added to $\sC$. Since $c$ is an uncovered point, its distance to any other center $c'$ that is already in $\sC$ is more than $2\alpha\cdot r(c) = 2\alpha \cdot \max\{r(c),r(c')\}$ where the equality follows from the fact that centers are picked in a non-decreasing order of their fair radius. Hence, for any pair of centers in $\sC$, property~\eqref{def:critBalls:prop2} holds.

Finally, by property~\eqref{def:critBalls:prop2}, balls of radius $r(.)$ around the centers present in $\sC$ are disjoint. Moreover, by the definition of fair radius, each of the balls $\{B(c, r(c))\}_{c\in \sC}$ contains at least $n/k$ points. Hence, the number of critical regions is at most $k$.
\end{proof}


As shown in~\citep{mahabadi2020individual}, the benefit of a set of critical regions is that it reduces the problem of finding an $\alpha$-fair clustering to a clustering problem with lower bound requirements, i.e., at least one center must be selected from each critical region.
We say that a set of cluster centers $C$ is \textit{feasible} w.r.t.~a set of critical regions $\sB$, if for every ball $B\in \sB$, $|C\cap B|>0$.


\begin{lemma}\label{lem:fair-sol-critical}
Let $\sB=\{B(c_1,\alpha\cdot r(c_1)),\ldots,B(c_m,\alpha\cdot r(c_m))\}$ be a set of critical areas obtained from Algorithm~\ref{alg:AlphaCluster} for a set of points $P$ with parameters $k$ and $\alpha$.
Then, any set of centers $S$ that is feasible w.r.t.~$\sB$ is $(3\alpha)$-fair.
\end{lemma}
\begin{proof}
Let $S$ be a set of cluster centers that is feasible w.r.t.~$\sB$. 
For every point $x\in P$ let $c_x$ denote the first center picked by Algorithm~\ref{alg:AlphaCluster} such that $x\in B(c_x, \alpha\cdot r(c_x))$. Moreover, let $s_x$ denote the center in $S$ such that $s_x \in B(c_x, \alpha\cdot r(c_x))$. Then, for any point $x\in P$:
\begin{align*}
 d(x, s_x) 
 \leq d(x, c_x) + d(c_x, s_x) 
 \leq 2\alpha\cdot r(x) + d(c_x, s_x)
 \leq 2\alpha\cdot r(x) + \alpha\cdot r(c_x)
 \leq 3\alpha\cdot r(x),
\end{align*}
where the first inequality follows from the triangle inequality, the second inequality follows from the property~\eqref{def:critBalls:prop1} of critical regions, the third inequality follows since $s_x\in B(c_x, \alpha\cdot r(c_x))$ and  the last inequality follows since centers are added in a non-decreasing order of their fair radius in line~\ref{alg:AlphaCluster:NonDecreasingR} of Algorithm~\ref{alg:AlphaCluster}; $r(c_x) \leq r(x)$.
\end{proof}

\paragraph{Facility location under matroid constraint.} Now we formally define the facility location problem {\em with $\ell_p$-norm cost} under matroid constraint to which we reduce the problem of fair clustering with $\ell_p$-norm cost. We remark that for our application, it suffices to solve the facility location under {\em partition matroid} constraint.

In facility location with $\ell_p$-norm cost, we are given a set of facilities $\sF$ and a set of clients $\sD$ where each facility $u$ has an opening cost of $f(u)$ and each client $v$ is assigned with a weight (or demand) $w(v)$. The cost of assigning one unit of weight (or demand) of client $v$ to facility $u$ is $d(v,u)^p$. Furthermore, we are given a matroid $\sM=(\sF, \sI)$. Then the goal is to choose a set of facilities $F$ that forms an independent set in $\sM$ and minimizes the total facility opening and client assignment cost. Formally,
\begin{align}\label{def:facility-location-matroid}
\argmin_{F\in\sI}\sum_{u \in F} f(u)+ \sum_{v\in \sD} w(v) \cdot d(v,\sF)^p
\end{align}

Next, we show a reduction from the $\alpha$-fair $k$-clustering problem to the facility location problem under matroid constraint. Then, in Section~\ref{sec:facility-lp-approx}, we generalize the result of~\cite{swamy2016improved} and devise an approximation algorithm for facility location with $\ell_p$-norm cost under matroid constraint.

\paragraph{Reduction to facility location under matroid constraint.}
Consider an instance of $\alpha$-fair $k$-clustering on a set of points $P$. Let $\sB$ be the set of critical regions of $P$ with parameters $k$ and $\alpha$ constructed via Algorithm~\ref{alg:AlphaCluster}. Then, given an instance of $\alpha$-fair $k$-clustering, Algorithm~\ref{alg:reduction} constructs an instance of facility location problem under matroid constraint. Before stating the main reduction, we show that the distance function $d'$ constructed in Algorithm~\ref{alg:reduction} is a metric distance.
\begin{lemma}\label{lem:d'-metric}
The distance function $d':(\sF\cup \sM)\times (\sF\cup \sM) \rightarrow \mathbb{R}^+$ as constructed in Algorithm~\ref{alg:reduction} constitutes a metric space.
\end{lemma}
\begin{algorithm}[h]
	\begin{algorithmic}[1]
		\STATE {\bfseries Input:} set of points $P$, target number of centers $k$, fairness parameter $\alpha$, accuracy parameter $\eps<1$, approximation guarantee of facility location under matroid constraint with $\ell_p$-norm cost $\beta \ge 1$ 
		\STATE \textbf{compute} a set of critical regions $\sB = \{B_1, \cdots, B_m\}$ via Algorithm~\ref{alg:AlphaCluster} on $(P, k, \alpha)$
		\STATE\COMMENT{Construction of facilities}
		\STATE {\bf let} $P_F = \{v_f \sep v\in P\}$ be a copy of $P$  
		\STATE{\bf let} $B_{F,i} = \{v_{f,i} \sep v\in B_i\}$ be a copy of $B_i$ {\bf for all} $B_i\in \sB$
		\STATE $\sF \leftarrow P_F \cup (\bigcup_{B_i \in \sB} B_{F, i})$ \quad \COMMENT{\emph{$\sF$ has two distinct copies of the points that belongs to a critical ball of $\sB$.}}
		
		\STATE $f(u) = 0$ \textbf{for all} $u\in \sF$ 
		\STATE\COMMENT{Construction of facilities}
		\STATE {\bf let} $\sD = \{v_c \sep v\in P\}$ be a copy of $P$ \quad \COMMENT{\emph{$\sD$ is a distinct copy of $P$.}} \label{alg:line:clients-same-points}
		\STATE $w(v) = 1$ \textbf{for all} $v\in \sD$
		
		
		\STATE\COMMENT{Construction of distance function $d':(\sF\cup \sD)\times (\sF\cup \sD)\rightarrow \mathbb{R}^+$}
		\STATE {\bf let} $\delta \leftarrow  \min_{x,y\in P}d(x,y)$
		\STATE{\bf let} $d'(u,u) = 0$ {\bf for all} $u\in \sF \cup \sD$ \label{alg:reduction:zero}
		\STATE{\bf let} $d'(v_x,u_y) = d(v,u)$ {\bf for all} $v_x, u_y \in \sF \cup \sD$ where $v\neq u$ \label{alg:reduction:same}
		\STATE{\bf let} $d'(v_x, v_y) = \min\{(\frac{\eps(n-k)}{\beta \cdot k})^{1/p},1\}\cdot \delta$ {\bf for all} $v_x, v_y\in \sF\cup \sD$ \label{alg:reduction:epsilon} 
		\STATE\COMMENT{Construction of matroid $\sM$}
		\STATE $\sM \leftarrow$ partition matroid s.t. $|I \cap B_{F, i}|\leq 1$ \textbf{for all} $i\in [m]$ and $|I\cap P_F|\leq k-m$ 
		\RETURN $(\sF, \sD, d', \sM)$
    \end{algorithmic}
	\caption{outputs an instance of facility location under matroid constraint corresponding to the given instance of $\alpha$-fair $k$-clustering.}
	\label{alg:reduction}
\end{algorithm}

\begin{theorem}\label{thm:main-reduction}
Suppose that there exists a $\beta$-approximation algorithm for the facility location with $\ell_p$-norm cost under matroid constraint. Then, for any $\eps>0$, there exists a $(\beta+\eps, 3)$-bicriteria approximation for $\alpha$-fair $k$-clustering with $\ell_p$-norm cost.
\end{theorem}
\begin{proof}
Let $\clla$ be a $\beta$-approximation algorithm for facility location with $\ell_p$-norm cost under matroid constraint. Consider an instance of $\alpha$-fair $k$-clustering with $\ell_p$-norm cost on pointset $P$ and let $(\sF, \sD, d', \sM)$ be the instance of facility location constructed by Algorithm~\ref{alg:reduction} with input parameters $P, k$ and $\alpha$. We show that the solution returned by $\clla(\sF, \sD, d', \sM)$ can be converted to a $(\beta+\eps,3)$-bicriteria approximation for the given instance of $\alpha$-fair $k$-clustering on $P$.

Let $\sB= \{B_1, \cdots, B_m\}$ be the critical regions constructed in Algorithm~\ref{alg:reduction}.
Let $\sol_F$ be the solution returned by $\clla(\sF, \sD, d', \sM)$ and let $\opt$ be an optimal solution of $\alpha$-fair $k$-clustering of $P$. Note that since adding centers to $\sol_F$ only reduces the $\ell_p$-cost of the solution on $(\sF, \sD, d', \sM)$, without loss of generality we can assume that $\sol_F$ picks exactly one center from each of $B_{F, i}$, for $i\in [m]$, and exactly $k-m$ centers from $P_F$.
Now we construct a solution $\sol$ of $k$-clustering on $P$ using the solution $\sol_F$. We start with an initially empty set of centers $\sol$. For each $B_i\in \sB$, let $c_{f,i}$ denote the center in $\sol_F \cap B_{F,i}$. In the first step, we add the point $c\in P$ corresponding to $c_{f,i}$ to $\sol$. Next, in the second step, for each $o_f\in  \sol_F \cap P_F$, we add the point $o\in  P$ corresponding to $o_{f}$ to $\sol$. Note that as some of these points may have already been added to $\sol$ in the first step, the final solution has at most $k$ distinct centers.

\paragraph{Fairness approximation.} 
By the first step in the construction of $\sol$ (from the given solution $\sol_F$), for each $i\in [m]$, $|B_i \cap \sol|\geq 1$. Hence, by Lemma~\ref{lem:fair-sol-critical}, $\sol$ is a $(3\alpha)$-fair clustering of $P$.
    
\paragraph{Cost approximation.} Note that by our reduction, all facility opening costs are set to zero. 
Next, we bound the assignment cost of points in $P$ to their closest centers in $\sol$ in terms of the the assignment cost of their corresponding client in $\sD$ to their closest facility in $\sF$: if 
$v\notin \sol$, $d'(v_c, \sol_F) = d(v, \sol)$; otherwise, $d'(v_c, \sol_F) = \min\{(\frac{\eps(n-k)}{\beta \cdot k})^{1/p},1\}\cdot \delta > 0 = d(v, \sol)$.
Hence, 
\begin{align}\label{eq:cost-sol-ub}
    \cost(P, \sol;p) \leq \cost(\sD, \sol_F; p).
\end{align}

Next, we bound the $\ell_p$-cost of $\sol_F$ on $(\sF, \sD, d', \sM)$ in terms of the optimal $\ell_p$-cost of fair clustering $P$---the cost of clustering $P$ using $\opt$. By the definition of $\alpha$-fairness, each point $v\in P$ must have a center within distance at most $\alpha \cdot r(v)$. Hence, for each critical region $B\in \sB$, $|\opt\cap B|\geq 1$. For each $i\in [m]$, let $c_{f, i} \in \sF$ be the copy of an arbitrary center $c_i\in \opt\cap B_i$ in the set $B_{F,i}$. For the remaining points in $\opt$, we pick their corresponding copies in the set $P_F$; the corresponding facilities of type $c_f$. Let $\opt_{F}$ denote the constructed solution for the instance $(\sF, \sD, d', \sM)$. 
Since $\opt_F$ picks exactly one point from each set $B_{F,i}$, for $i\in [m]$, and exactly $k-m$ points from $P_F$, $\opt_F$ is a feasible solution of instance $(\sF, \sD, d', \sM)$. Moreover, since all facility opening cost are set to zero, the $\ell_p$-clustering cost of $\opt$ on pointset $P$ is
    \begin{align*}
        \cost(P, \opt; p) &= \sum_{v\in P} d(v, \opt)^p \\
        &= \sum_{v\in \opt} d(v, \opt)^p + \sum_{v\in P\setminus \opt} d(v, \opt)^p \\
        &= \sum_{v\in \opt} \big(d'(v_c, \opt_F)^p - (\frac{\eps(n-k)}{\beta\cdot k}) \cdot \delta^p \big) + \sum_{v\in P\setminus \opt} d'(v_c, \opt_F)^p \\
        &\geq\sum_{v_c\in \sD} d'(v_c, \opt_F)^p - k \cdot (\frac{\eps(n-k)}{\beta \cdot k}) \cdot \delta^p \\
        &= \cost(\sD, \opt_F; p) - \frac{\eps}{\beta}(n-k)\cdot \delta^p \\
        &\geq \cost(\sD, \opt_F; p) - \frac{\eps}{\beta} \cdot \cost(P, \opt; p), 
    \end{align*}
where the last inequality holds since $\cost(P, \opt; p) \geq (n-k)\cdot \delta^p$.    
Hence, 
\begin{align}\label{eq:cost-opt-lb}
    \cost(\sD, \opt_F; p) \leq (1+\frac{\eps}{\beta})\cost(P, \opt;p).
\end{align}
Thus,
\begin{align*}
    \cost(P, \sol; p) 
    &\leq \cost(\sD, \sol_F; p)\;\rhd\text{by~\eqref{eq:cost-sol-ub}} \\
    &\leq \beta\cdot \cost(\sD, \opt_F; p) \\
    &\leq (\beta+\eps)\cdot \cost(P, \opt; p). \;\rhd\text{by~\eqref{eq:cost-opt-lb}}
\end{align*}

In other words, the $\ell_p$-cost of clustering $P$ using $\sol$ is within a $\beta+\eps$ factor of the optimal $\alpha$-fair $k$-clustering of $P$.  
\end{proof}

\begin{proofof}{\bf Theorem~\ref{thm:main}.}
The proof follows from Theorem~\ref{thm:main-reduction} and the $16^p$-approximation algorithm of facility location with $\ell_p$-norm cost under matroid constraint for $p>1$ shown in Theorem~\ref{thm:main-facility}.

Next, for the case with $p=1$, which corresponds to $k$-median, we can employ the approximation guarantee of~\citep{krishnaswamy2018constant} and achieve a better cost approximation factor. In this setting, the proof follows from Theorem~\ref{thm:main-reduction} and the $7.081$-approximation algorithm of~\citep{krishnaswamy2018constant} for facility location under matroid constraint.
\end{proofof}

\subsection{A Simpler Reduction of Fair $k$-Center}
Here, we show a reduction of the $\alpha$-fair $k$-center problem to the $k$-center problem under partition matroid constraint. Then, exploiting the $3$-approximation algorithm of~\cite{jones2020fair}, we achieve a better approximation guarantee for the $\alpha$-fair $k$-center problem. 

Consider an instance of $\alpha$-fair $k$-center on a pointset $P$. Let $\sB$ be the set of critical regions of $P$ with parameters $k$ and $\alpha$ constructed via Algorithm~\ref{alg:AlphaCluster}. Then, given an instance of $\alpha$-fair $k$-center, Algorithm~\ref{alg:k-center-reduction} constructs an instance of $k$-center under partition matroid constraint. 
\begin{algorithm}[h]
	\begin{algorithmic}[1]
		\STATE {\bfseries Input:} set of points $P$, target number of centers $k$, fairness parameter $\alpha$, accuracy parameter $\eps <1/2$, approximation guarantee of $k$-center under partition matroid constraint $\beta \ge 1$
		\STATE \textbf{compute} a set of critical regions $\sB = \{B_1, \cdots, B_m\}$ via Algorithm~\ref{alg:AlphaCluster} on $(P, k, \alpha)$
		\STATE {\bf let} $\bar{P}_0 = \{v_0 \sep v\in P\}$ be a copy of $P$ 
		\STATE {\bf let} $\bar{B}_{i} =\{v_{i} | v\in B_i\}$ be a copy of $B_i$ for all $B_i\in \sB$, 
		\STATE $P' \leftarrow \bar{P}_0 \cup (\bigcup_{B_i \in \sB} \bar{B}_{i})$ \quad \COMMENT{\emph{$P'$ has two distinct copies of the points that belongs to a critical ball of $\sB$.}}
		\STATE $k_i = 1$ \textbf{for all} $i\in [m]$ \COMMENT{\emph{denotes that we pick at most one center from each critical ball.}}
		\STATE $k_0 = k - m$
        \STATE\COMMENT{Construction of distance function $d':P'\times P'\rightarrow \mathbb{R}^+$}
		\STATE {\bf let} $\delta \leftarrow  \min_{x,y\in P}d(x,y)$
		\STATE{\bf let} $d'(u,u) = 0$ {\bf for all} $u\in P'$ 
		\STATE{\bf let} $d'(v_x,u_y) = d(v,u)$ {\bf for all} $v_x, u_y \in P'$ where $v\neq u$ 
		\STATE{\bf let} $d'(v_x, v_y) = \eps\cdot\delta/\beta$ {\bf for all} $v_x, v_y\in P'$ 
        \RETURN $(P', \{(\bar{P}_0,k_0), (\bar{B}_1, k_1), \cdots, (\bar{B}_m, k_m)\}, d')$
    \end{algorithmic}
	\caption{outputs an instance of $k$-center under partition matroid constraint corresponding to the given instance of $\alpha$-fair $k$-center.}
	\label{alg:k-center-reduction}
\end{algorithm}
Similarly to Lemma~\ref{lem:d'-metric}, we first show that the distance function $d'$ constructed in Algorithm~\ref{alg:k-center-reduction} is a metric distance.
\begin{lemma}\label{lem:d'-kcenter-metric}
The distance function $d':P'\times P' \rightarrow \mathbb{R}^+$ as constructed in Algorithm~\ref{alg:k-center-reduction} constitutes a metric space.
\end{lemma}

\begin{theorem}\label{thm:k-center-reduction}
Suppose that there exists a $\beta$-approximation algorithm for $k$-center under partition matroid constraint. Then, there exists a $(\beta+\eps, 3)$-bicriteria approximation for $\alpha$-fair $k$-center.
\end{theorem}

\begin{theorem}\label{thm:fair-k-center}
For any $\alpha\geq 1$, there exists a polynomial time algorithm that computes a $(3+\eps, 3)$-bicriteria approximate solution for $\alpha$-fair $k$-center.
\end{theorem}
\begin{proof}
The proof follows from Theorem~\ref{thm:k-center-reduction} and the $3$-approximation algorithm of ~\citep{jones2020fair} for $k$-center under partition matroid constraint.
\end{proof}

\subsubsection*{Acknowledgments}
The first author thanks Benjamin Moseley and Rudy Zhou for their helpful feedback on implications of~\cite{gupta2021structural} in our setting.
The second author thanks Beate Bollig for her helpful comments on the exposition of the paper.
Finally, the authors thank anonymous reviewers for their detailed comments which helped improve the paper. 
\bibliographystyle{abbrvnat}
\bibliography{fair-cls}
\newpage
\onecolumn
\appendix
\section{\texorpdfstring{$e^{O(p)}$}{}-Approximation for Matroid
Facility Location with \texorpdfstring{$\ell_p$}{}-Cost}\label{sec:facility-lp-approx}
We consider a natural LP-relaxation for the facility location problem under matroid constraint considered by~\cite{swamy2016improved}. This relaxation is a generalization of the standard LP relaxation of $k$-median and $k$-means clustering (e.g. ~\citep{charikar2002constant}). For every facility $u\in \sF$, the variable $y_u$ denotes whether the facility $u$ is open or not; $y_u=1$ if $u$ is open and $y_u=0$ otherwise. For every client $v\in P$ and facility $u\in \sF$, we have a variable $x_{vu}$ that denotes whether $u$ is the closest facility among the opened facilities to $v$. We also assume that we are given a function $w: P\rightarrow \mathbb{R}$ which denotes the demand of clients.
Finally, we use $r$ to denote the rank function of the matroid $\sM = (\sF, \sI)$.
\begin{align}
\textbf{LP Relaxation: }&\cllp(w, \sM)\nonumber\\[1mm]
\text{minimize }& \ \rlap{$\sum_{u \in \sF} f(u) \cdot y_u + \sum_{v\in P, u\in \sF} w(v) \cdot d(v, u)^p \cdot x_{vu}$} \nonumber\\[1mm]
\text{s.t.}\qquad &\sum_{u\in \sF} x_{vu} \geq 1 &&\forall v\in P \label{casgn} \\ 
&\sum_{u\in S} y_u \leq r(S) &&\forall S\subseteq \sF  \label{matroid_constraint}\\
&0 \leq x_{uv} \leq y_u  &&\forall u\in \sF,v\in P \label{under_assignment}
\end{align}

Note that for an optimal solution $(x, y)$ of $\cllp$, may assume that for every client $v$, $\sum_{u\in \sF} x_{vu}=1$.

We follow the framework of \citep{charikar2002constant,swamy2016improved}. First, we reduce the instance into a \textit{well-separated} instance and then we find a half-integral solution of the well-separated instance. 
We start with an optimal \textit{fractional} solution $(x^*,y^*)$ to $\cllp(w,\sM)$. The cost of this optimal fractional solution is denoted by $z^*$. 
The ultimate goal is to obtain a good integral solution whose cost is comparable to $z^*$.
First we construct a modified instance, called \textit{well-separated} instance. A key property of the well-separated instance is that the distance of any pair of clients $v$ and $u$ in this instance is ``large" compared to the contributions of each of $u$ and $v$ w.r.t. $(x^*, y^*)$.  
Then, we prove two important statements: (1) there exists a fractional solution of the well-separated instance whose cost is not more than $z^*$ and (2) an integral solution of the well-separated instance with cost $z'$ can be transformed into a solution of cost $e^{O(p)} (z' + z^*)$ on the original instance.  

Hence, it suffices to find an \textit{integral} solution $F'$ of the well-separated instance whose cost is not ``much larger'' than $z^*$. If we achieve that, then we get a ``good'' approximate integral solution of the original instance.


In the first step, we construct a half-integral solution $\hat{y}$ to the well-separated instance whose cost is not more than $3^p\cdot z^*$. First, we show that there exists a solution $y'$ with ``certain structure'' whose cost is at most $3^p \cdot z^*$. Then, we consider a modified cost function that plays a role as a proxy for the actual cost and show that under the new cost function we can always find a feasible half-integral solution with minimum cost. Lastly, we show that the actual cost of the constructed half-integral solution is also at most $3^p \cdot z^*$. 

In the second step, we construct an integral solution $\tilde{y}$ to the well-separated instance from the half-integral solution $\hat{y}$ and show that the cost of the integral solution is at most $e^{O(p)}$ times the cost of the half-integral solution. Again, we rely on the integrality of matroid intersection polytopes to construct an approximately good integral solution from the given half-integral solution.

All together, we obtain an integral solution to the original instance of cost at most $e^{O(p)}\cdot z^*$.

\subsection{Obtaining a half-integral solution}\label{sec:half-integral}
%

\paragraph{Step I: Consolidating Clients.} \label{step1}
In this section, we analyze the ``client consolidation'' subroutine, Algorithm~\ref{alg:consolidating_clients}, which outputs a new set of demands that is supported on a well-separated set of clients. 

An important notion in the framework of~\citep{charikar2002constant} is the {\em fractional distance of a client $v$ to its facility} w.r.t. an optimal fractional solution $(x,y)$. In our setting with $\ell_p$ clustering cost, the fractional distance of clients is defined as $\sR(v) := \big(\sum_{u\in \sF} d(v, u)^p \cdot x_{vu}\big)^{1/p}$. In particular, if $(x,y)$ is an integral solution then $\sR(v)$ denotes the distance of client $v$ to the facility it is assigned to. In other words, $\sR(v)$ is the assignment cost of one unit of demand at client $v$. 

We consider the clients $v_1, \cdots, v_n$ in a non-decreasing order of their fractional distances: $\sR(v_1) \le \cdots \le \sR(v_n)$. At the time we are processing client $v_i$ with {\em non-zero} demand, we check whether there exist another client $v_j$ with non-zero demand such that $j > i$ and $d(v_i, v_j) \leq 2^{\frac{p+1}{p}} \cdot \sR(v_j)$. If there exists such a client, then we add the demand of $v_j$ to $v_i$ and set the demand of $v_j$ to zero.

Observe that throughout the client consolidation subroutine, once a client is processed, the algorithm never moves its demand in the rest of the procedure. This in particular implies that during the client consolidation subroutine, the demand of each client moves at most once.  

\begin{algorithm}[h]
	\begin{algorithmic}[1]
		\STATE {\bfseries Input:} $(x, y)$ is an optimal solution of $\cllp(w, \sM)$
		\STATE $\sR(v) = \left(\sum_{u\in \sF}  d(v,u)^p \cdot x_{vu}\right)^{1/p}$ \textbf{for all} $v\in P$	
        \STATE $w'(v) = w(v)$ \textbf{for all} clients $v\in P$
		\STATE {\bf sort} the points in $P$ so that $\R(v_1) \leq \R(v_2) \leq \cdots \leq \R(v_n)$
		\FOR{$i = 1$ to $n-1$}
			\FOR{$j = i+1$ to $n$}
				\IF{$d(v_i, v_j) \leq 2^{\frac{p+1}{p}} \R(v_j)$  and $w'(v_i) > 0$}
					\STATE $w'(v_i) = w'(v_i) + w'(v_j)$ 
                    \STATE $w'(v_j) = 0$
				\ENDIF
			\ENDFOR
		\ENDFOR	
    \end{algorithmic}
	\caption{consolidating clients.}
	\label{alg:consolidating_clients}
\end{algorithm}


Next we define $P'\subseteq P$ as the support of $w'$. The following claim shows that the clients in $P'$ are well-separated.
\begin{claim}\label{clm:well-separated}
For every pair of clients $u,v \in P'$, $d(u,v) > 2^{\frac{p+1}{p}}\max(\sR(v), \sR(u))$
\end{claim}
\begin{proof}
Suppose there exists a pair of clients $v, u\in P'$ such that $d(v,u) \leq  2^\frac{p+1}{p}\cdot \max (\sR(v),\sR(u))$. Wlog, assume that $\sR(v)\leq \sR(u)$.
However, since $d(v,u)\leq 2^\frac{p+1}{p}\cdot \sR(u)$, at the iteration in Algorithm~\ref{alg:consolidating_clients} that processes client $u$, the algorithm will move the demand of $u$ to $v$. Hence, at the end of the algorithm $w'(u) = 0$ which is a contradiction.
\end{proof}


Next, we show that a feasible solution of $\cllp$ for the original instance is a feasible solution of $\cllp$ for the constructed well-separated instance with the same or smaller cost.

\begin{lemma}\label{lem:feasible-sol}
Let $(x,y)$ be a feasible solution of $\cllp(w, \sM)$ with cost $z$. Then, $(x,y)$ is a feasible solution of $\cllp(w', \sM)$ with cost at most $z$.
\end{lemma}
\begin{proof}
Since the set of constraints in $\cllp(w', \sM)$ is the same as the set of constraints in $\cllp(w, \sM)$, $(x,y)$ is a feasible solution of $\cllp(w', \sM)$.
Moreover, in Algorithm~\ref{alg:consolidating_clients}, a client $u$ moves its demand to another client $v$ with a lower assignment cost (i.e., $\sR(v) \leq \sR(u))$. Hence, the cost of solution $(x,y)$ on $\cllp(w', \sM)$ is at most $z$.
\end{proof}

\begin{theorem}\label{thm:convert-to-original}
Let $F'$ be an integral solution of the well-separated instance $(w', \sM)$ of cost at most $z'$. Then, 
$F'$ is a solution of the original instance $(w, \sM)$ of cost at most $4 \cdot 16^{p-1}\cdot z^* + (\frac{8}{7})^{p-1}\cdot z'$ where $z^*$ is the optimal cost of $\cllp(w, \sM)$.
\end{theorem}
\begin{proof}
Since the set of constraints in $\cllp(w, \sM)$ is the same as the set of constraints in $\cllp(w', \sM)$, $F'$ is a feasible solution of $\cllp(w, \sM)$. 
For every client $v\in P$, we assume that Algorithm~\ref{alg:consolidating_clients} has moved the demand of $v$ to $v'\in P'$. More precisely, the algorithm may either move the demand of $v$ to another client $v' = u$ or keep it at the same client $v' = v$. Moreover, in both cases, $d(v, v') \leq 2^\frac{p+1}{p} \sR(v)$. Hence,
\begin{align}
    d(v, F')^p &\leq 8^{p-1}\cdot d(v, v')^p + (\frac{8}{7})^{p-1}\cdot d(v', F')^p &&\rhd\text{Observation~\ref{obser:gen-triangle-ineq} with $\lambda = 7$} \nonumber \\
    &\leq 4 \cdot 16^{p-1}\cdot \sR(v)^p + (\frac{8}{7})^{p-1}\cdot \sR(v')^p &&\rhd\text{since $d(v, v') \leq 2^\frac{p+1}{p} \sR(v)$}\label{eq:distance-F'}
\end{align}
Thus, the cost of solution $F'$ over the original instance is, 
\begin{align*}
    \sum_{u\in F'} f(u) + \sum_{v\in P} d(v, F')^p 
    &\leq  \sum_{u\in F'} f(u) + \sum_{v\in P} 4 \cdot 16^{p-1}\cdot \sR(v)^p + (\frac{8}{7})^{p-1}\cdot \sR(v')^p &&\rhd\text{by Eq.~\eqref{eq:distance-F'}}\nonumber \\
    &\leq 4\cdot 16^{p-1} \cdot z^* + (\frac{8}{7})^{p-1} \cdot z'
\end{align*}
\end{proof}

\paragraph{Step II: Transforming to a half-integral solution.}
In this section, we provide a method to construct a half-integral solution $(w', \sM)$ from the optimal fractional solution $(x, y)$, previously denoted as $(x^*,y^*)$. More precisely, we start with the optimal fractional solution $(x,y)$ and after two steps, construct a feasible half-integral solution of the well-separated instance whose cost is not more than $3^p \cdot z^*$.

First, we define a few useful notions for our algorithm in this section and its analysis. For every client $v\in P'$, we define $F(v)$ to be the set of all facilities $u$ such that $v$ is the closest client to $u$ in $P'$, i.e. $F(v):= \{u\in \sF: d(v,u) = \min_{s\in P'} d(s,u)\}$ with ties broken arbitrarily. 
Furthermore, let $F'(v)\subseteq F(v) := \{u \in F(v): d(u,v )\leq 2^{1/p}\cdot \sR(v)\}$. 
Lastly, for each client $v$, we define $\gamma_v := \min_{u\notin F(v)} d(v,u)$ and let $G(v) := \{u\in F(v): d(v,u)\leq \gamma_v\}$.

First, we show that for every client $v$, all ``nearby'' facilities are contained in $F(v)$. 
\begin{lemma}\label{lem:nearby-facilities}
For every client $v \in P'$, the set $F(v)$ contains all the facilities $u$ such that $d(v,u)\leq 2^{1/p} \cdot \sR(v)$.
\end{lemma}
\begin{proof}
Suppose that there exists a facility $u\notin F(v)$ with $d(v,u) \leq 2^{1/p} \cdot \sR(v)$. Hence, there exists a client $v'\in P'$ such that $d(v', u) \leq d(v, u) \leq 2^{1/p} \cdot \sR(v)$. By the approximate triangle inequality for $d(\cdot, \cdot)^p$,
\begin{align*}
 d(v, v')^p & \leq 2^{p-1}\cdot(d(v, u)^p + d(u, v')^p) \\
 & \leq 2^{p}\cdot d(v,u)^p &&\rhd\text{since $d(v', u) \leq d(v, u)$}\\
 & \leq 2^{p+1}\cdot \sR(v)^p,
\end{align*}
which contradicts the well-separatedness property of clients in $P'$.
\end{proof}
\begin{corollary}\label{cor:gamma}
For every client $v\in P'$, $\gamma_v > 2^{1/p} \cdot \sR(v)$. In particular, for every client $v \in P'$, $F'(v) \subseteq G(v)$.
\end{corollary}

\begin{proof}
Let $u\in \sF \setminus F(v)$ be the facility such that $\gamma_v=d(u,v)\leq 2^{1/p} \cdot \sR(v)$. Then, by Lemma~\ref{lem:nearby-facilities}, facility $u$ is element of $F(v)$, giving a contradiction. For the second part, consider a facility $u\in F'(v)$. Then, by definition of $F'$ and the first part of this proof, $d(u,v)\leq 2^{1/p} \cdot \sR(v) < \gamma_v$ and therefore $u\in G(v)$.
\end{proof}
For convenience, we provide a simple subroutine that computes the optimal assignments of clients (i.e. $x$) for any given feasible (possibly {\em fractional}) set of open facilities $y$. 
Note that if $y$ is integral (resp. half-integral), the resulting optimal assignment $x$ is integral (resp. half-integral) too.

\begin{algorithm}[h]
	\begin{algorithmic}[1]
		\STATE {\bfseries Input:} open facilities $y$ and demand function $w: P \rightarrow \mathbb{R}$
		\STATE \textbf{initialize} $x_{uv}\gets 0$ \textbf{for all} $u\in \sF$, $v\in P$. 
		\WHILE{there exists $v\in P$ with $w(v)>0$ and $\sum_{u\in \sF}x_{uv}<1$}
		\STATE {\bf find} a closest facility $u$ to $v$ that satisfies $y_u-x_{uv}>0$.
		\STATE $x_{uv}\gets \min(1-\sum_{u\in \sF}x_{uv},y_u)$.\label{algo:line:RespectProp}
		\ENDWHILE
		\RETURN $x$
    \end{algorithmic}
	\caption{constructs an optimal assignment of clients to a given set of facilities.}
	\label{alg:assign_optimal}
\end{algorithm}

Next, we show another useful property of $F(\cdot)$ which is crucial in constructing the half-integral solution of the well-separated instance.
\begin{claim}\label{claim:markov}
For every client $v\in P'$, $\sum_{u \in F'(v)} x_{vu} \geq 1/2$.
\end{claim}
\begin{proof}
To prove the statement, we show that $\sum_{u \notin F'(v)} x_{vu} \leq 1/2$.
\begin{align*}
    \sum_{u \notin F'(v)} x_{vu} \cdot 2\cdot \sR(v)^p 
    &\leq \sum_{u \notin F'(v)} x_{vu} \cdot d(u,v)^p &&\rhd\text{by Corollary~\ref{cor:gamma}} \\
    &\leq \sum_{u \in \sF} x_{vu} \cdot d(u,v)^p = \sR(v)^p
\end{align*}
Hence, $\sum_{u \notin F'(v)} x_{vu} \leq 1/2$. Since in an optimal solution of $\cllp(w, \sM)$, for every client $v$, $\sum_{u\in \sF} x_{vu} =1$. Thus, we have that $\sum_{u \in F'(v)} x_{vu} \geq 1/2$. 
\end{proof}
%

Next, we show that in the well-separated instance, for every client $v\in P'$ there exists another client $v'\in P'$ whose ``nearby'' facilities are close to $v$ as well. More precisely,
\begin{claim}\label{clm:3gamma}
Consider a client $v\in P'$. Let $u\in F(v')$ be the facility such that $\gamma_v = d(v,u)$. Then, for every facility $u' \in F'(v')$, $d(v, u')\le 3\gamma_v$.
\end{claim}

\begin{proof}
See Figure~\ref{fig:scales} for an illustration of the relevant distances. First, we bound $d(v, v')$.
\begin{align}
 d(v, v')^p 
 &\leq 2^{p-1}\cdot(d(v, u)^p + d(u, v')^p) &&\rhd \text{approximate triangle inequality} \nonumber\\
 & \leq 2^{p}\cdot d(v, u)^p &&\rhd \text{since $u\in F(v')$} \nonumber\\
 & \leq 2^{p}\cdot \gamma_v^p \label{eq:clients-dist}
\end{align}
Furthermore, since clients in $P'$ are well-separated, by Claim~\ref{clm:well-separated}, $d(v, v')^p > 2^{p+1} \cdot \max(\sR(v),\sR(v'))^p$. 
Moreover, since $u' \in F'(v')$, by Lemma~\ref{lem:nearby-facilities}, $d(v', u')^p \leq 2 \cdot \sR(v')^p$. Hence,
\begin{align}
    d(v', u')^p 
    &\leq 2\cdot \sR(v')^p &&\rhd\text{by Lemma~\ref{lem:nearby-facilities}} \nonumber\\
    &\leq 2\cdot \max(\sR(v),\sR(v'))^p \nonumber\\
    &\leq \frac{1}{2^p} \cdot d(v, v')^p &&\rhd\text{by the well-separateness property of $P'$ (Claim~\ref{clm:well-separated})}\nonumber \\
    &\leq \gamma_v^p &&\rhd\text{Eq.~\eqref{eq:clients-dist}}\label{eq:client-facility-dist}
\end{align}
By an application of the general form of approximate triangle inequality for $d(\cdot, \cdot)^p$,
\begin{align*}
 d(v, u')^p 
 &\leq (\frac{3}{2})^{p-1}\cdot d(v, v')^p + 3^{p-1}\cdot d(v', u') &&\rhd\text{Corollary~\ref{cor:gen-triangle-ineq} with $\lambda =2$}\\
 &\leq 2\cdot 3^{p-1}\cdot \gamma_v^p + 3^{p-1}\cdot \gamma_v^p &&\rhd\text{Eq.~\eqref{eq:clients-dist} and~\eqref{eq:client-facility-dist}} \\
 &\leq 3^p \cdot \gamma_v^p 
\end{align*}



\end{proof}
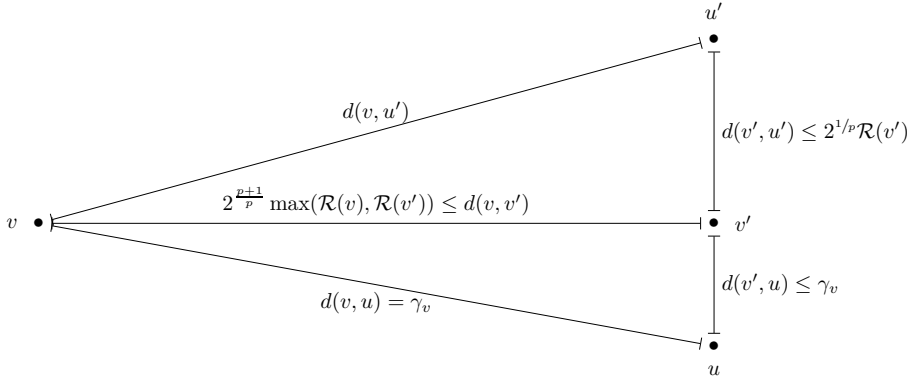
\begin{figure}
	\centering
	\resizebox{.75\textwidth}{!}{
	\begin{tikzpicture}[x=1cm,y=1cm]
	
	\node (v) [black,label=left:{$v$}] at (1,3) {\textbullet};
	\node (v') [black,label=right:{$v'$}] at (12,3) {\textbullet};
	\draw [|-|] (v) to node [above]{$2^{\frac{p+1}{p}} \max(\sR(v),\sR(v'))\le d(v,v')$}(v');
	\node (u) [black,label=below:{$u$}] at (12,1) {\textbullet};
	\draw [|-|] (v') to node [right]{$d(v',u)\le \gamma_v$}(u);
	\draw [|-|] (v) to node [below]{$d(v,u)=\gamma_v$}(u);
	\node (u') [black,label=above:{$u'$}] at (12,6) {\textbullet};
	\draw [|-|] (v') to node [right]{$d(v',u')\le 2^{\sfrac{1}{p}}\sR(v')$}(u');
	\draw [|-|] (v) to node [above]{$d(v,u')$}(u');
	\end{tikzpicture}
	}
    \caption{Illustration of the relevant distances to bound the distance between $v$ and $u'$.} 
	\label{fig:scales}
\end{figure}

Next, we prove the main theorem of this section. 

\begin{theorem}\label{thm:main-half-integral}
Let $z^*$ denote the cost of an optimal solution of $\cllp(w, \sM)$. There exists a half-integral solution of $\cllp(w', \sM)$ of cost at most $3^p\cdot z^*$.
\end{theorem}

We begin with an optimal solution $(x,y)$ of $\cllp(w, \sM)$ of cost $z^*$ which by Lemma~\ref{lem:feasible-sol} is a feasible solution of $\cllp(w', \sM)$ of cost at most $z^*$. The first step in the proof is to construct an ``intermediate'' {\em feasible} solution $(x', y')$ as follows. 
$$
y'_u :=
\begin{cases}
      x_{vu} &\text{if there exists a client }  v\in P' \text{ such that } u\in G(v)\\
      0 &\text{otherwise}
\end{cases} 
$$
Note that if there is no $v\in P'$ such that $u\in G(v)$, then $y_u = 0$. 
\begin{claim}\label{clm:disjoint-G}
For every client $v\in P'$, $\sum_{u\in G(v)} y'_{u} \leq 1$.
\end{claim}
\begin{proof}
Note that since the sets of facilities $\{F(v)\}_{v\in P'}$ are disjoint and for each $v\in P'$, $G(v) \subseteq F(v)$, the sets $\{G(v)\}_{v\in P'}$ are disjoint too.
Thus, by the way we constructed the solution $y'$, 
\begin{align*}
    \sum_{u\in G(v)} y'_v = \sum_{u\in G(v)} x_{vu} \leq \sum_{u\in \sF} x_{vu} =1,
\end{align*}
where the first equality follows from the disjointness of $\{G(v)\}_{v\in P'}$ and the second equality follows from the optimality of the solution $(x,y)$.
\end{proof}


Next, following the approach of~\citep{swamy2016improved}, we introduce a {\em modified cost function} $\cst$ that serves as a proxy for the actual cost: 
\begin{align}\label{eq:proxy-cost}
    \cst(y') = \sum_{u\in \sF} y'_u \cdot f(u) + \sum_{v\in P'} w'(v) \cdot \Big(\sum_{u\in G(v)} d(v, u)^p \cdot y'_{u} + 3^p \gamma^p (1 - \sum_{u\in G(v)} y'_u)\Big)
\end{align}
This is crucial in bounding the cost of $(x', y')$, where $x'$ is the optimal assignment w.r.t. $y'$ and $w'$ as constructed by Algorithm~\ref{alg:assign_optimal}. Furthermore, the cost function $\cst$ plays an important role in showing the existence of a good approximate half-integral solution of $\cllp(w', \sM)$.

\begin{lemma}\label{lem:feasible-intermediate-sol}
The solution $(x', y')$ is a feasible solution of $\cllp(w', \sM)$ and $\cost(x',y') \leq \cst(y') \leq 3^p \cdot z^*$.
\end{lemma}
\begin{proof}
For every facility $u\in \bigcup_{v\in P'} G(v)$, let $v$ denote the client such that $u\in G(v)$. Note that the disjointness of $\{G(v)\}_{v\in P'}$ implies the uniqueness of such client.
Note that for every $S\subseteq \sF$, 
\begin{align*}
    \sum_{u\in S} y'_u 
    &= \sum_{u\in \bigcup_{v\in P'}G(v) \cap S} y'_u + \sum_{u\in S\setminus \bigcup_{v\in P'}G(v)} y'_u\\
    &= \sum_{u\in \bigcup_{v\in P'}G(v) \cap S} x_{vu} &&\rhd\text{by the definition of $y'$}\\
    &\leq \sum_{u\in S} y_u \leq r(S), 
\end{align*} 
where the last two inequalities hold since $(x,y)$ is a feasible solution of $\cllp(w', \sM)$. Hence, $(x',y')$ satisfies the matroid constraint $\sM$ (i.e., constraint~\eqref{matroid_constraint} in $\cllp$).

Furthermore, line~\ref{algo:line:RespectProp} of Algorithm~\ref{alg:assign_optimal} ensures the constraints~\eqref{casgn} and~\eqref{under_assignment} of $\cllp(w',\sM)$ are satisfied by $(x', y')$. Hence, $(x',y')$ is a feasible solution of $\cllp(w' ,\sM)$. 

To prove $\cost(x',y') \leq \cst(y')$ we show that $y'$ is contained in the polytope $\solsp$---the polytope $\solsp$ is defined formally in~\eqref{eq:sol-space}. Therefore, by Lemma~\ref{lem:proxy-cost-sol-space}, $\cst(y')\geq \cost(x', y')$.
The first condition of $\solsp$ encodes the matroid independence constraint with $\sum_{u\in S}y'_u \leq r(S)$. As already stated in the context of the feasibility of $(x', y')$ for $\cllp(w', \sM)$, $y'$ satisfies the matroid constraint.
Secondly, for each client $v \in P'$,
\begin{align*}
    \sum_{u\in F'(v)} y'_u 
    &= \sum_{u\in F'(v)} x_{vu} &&\rhd\text{by the definition of $y'$ and since $\forall v\in P', F'(v)\subseteq G(v)$} \\
    &\geq 1/2 &&\rhd\text{by Claim~\ref{claim:markov}}
\end{align*}
Finally, by Claim~\ref{clm:disjoint-G}, for every $v\in P'$, $\sum_{u\in G(v)}y'_u \leq 1$. Hence, $y' \in \solsp$.

Next, we show that $\cst(y') \leq 3^p \cdot z^*$.
\begin{align}
    \cst(y')
    &=\sum_{u\in \sF}f(u)\cdot y'_u + \sum_{v\in P'} w'(v) \cdot \big(\sum_{u\in G(v)}d(v,u)^p\cdot y'_u + 3^p\gamma^p_v (1-\sum_{u\in G(v)}y'_u)\big) \nonumber \\
    & \leq \sum_{u\in \sF}f(u)\cdot y'_u + \sum_{v\in P'}w'(v) \cdot \big(\sum_{u\in G(v)}d(v,u)^p \cdot x_{vu}+ 3^p \sum_{u\in \sF \setminus G(v)}d(v,u)^p \cdot x_{vu}\big) \label{ineq:F'-G}\\
    &\leq \sum_{u\in \sF}f(u)\cdot y_u+3^p\sum_{v\in P'}w'(v)\sum_{u\in \sF}d(v,u)^p\cdot x_{uv} \label{ineq:y'-y} \\
    & \leq 3^p\cdot z^* \nonumber &&\rhd\text{by Lemma~\ref{lem:feasible-sol}}
\end{align}
Inequality~\eqref{ineq:F'-G} holds since by the definition of $G(v)$, for every client $v\in P'$ and facility $u\in \sF\setminus G(v)$, $d(v,u) > \gamma_v$. Inequality~\eqref{ineq:y'-y} holds since for any non-zero $y'_u$ there exist a client $v$ such that $y'_u \leq x_{vu}\leq y_u$---note that if $y'_u=0$ then $y'_u\leq y_u$ holds trivially. 
\end{proof}
Besides the fact that the modified cost function $\cst$ provides an upper bound for the actual cost of {\em certain solutions} of the well-separated instance, by the standard results for the matroid intersection problem, we can find a half-integral solution $(x'', y'')$ whose modified cost $\cst$ is minimized. In particular, we can find a half-integral feasible solution $(x'', y'')$ of the well-separated instance such that $\cst(y'') \leq \cst(y')$. Before describing our method for constructing $(x'', y'')$, we formally define the set of solutions from which we choose the half-integral solution $(x'', y'')$.
\begin{align}\label{eq:sol-space}
\solsp := \{y \in \mathbb{R}^\mathcal{F}_+:\sum_{u\in S}y_u \leq r(S) \quad \forall S\subseteq \sF,\quad 1/2 \leq \sum_{u\in F'(v)} y_u, \sum_{u\in G(v)} y_u\leq 1 \quad \forall v \in P'\}
\end{align}
Note that $y'\in \solsp$ (it is formally proved in the proof of Lemma~\ref{lem:feasible-intermediate-sol}).
First, we show that for any solution $\bar{y} \in \solsp$ and its optimal assignment $\bar{x}$ w.r.t. $\bar{y}$ and $w'$ (e.g. as described in Algorithm~\ref{alg:assign_optimal}), $\cost(\bar{x}, \bar{y}) \leq \cst(\bar{y})$.

\begin{claim}\label{clm:feasible-polytope}
For every feasible solution $\bar{y}\in \solsp$ and any feasible assignment $\tilde{x}$ w.r.t. $\bar{y}$ and $w'$, the solution $(\tilde{x}, \bar{y})$ is a feasible solution of $\cllp(w', \sM)$. 
\end{claim}
\begin{proof}
Since $\bar{y}\in \solsp$, it trivially satisfies the matroid constraint $\sM$ (i.e., constraint~\eqref{matroid_constraint} in $\cllp$).
Furthermore, given that $\tilde{x}$ is a feasible assignment w.r.t. $\bar{y}$ and $w'$, $(\tilde{x}, \bar{y})$ satisfies constraints~\eqref{casgn} and~\eqref{under_assignment} of $\cllp(w',\sM)$. Hence, $(\tilde{x}, \bar{y})$ is a feasible solution of $\cllp(w' ,\sM)$. 
\end{proof}

\begin{lemma}\label{lem:proxy-cost-sol-space}
For every solution $\bar{y} \in \solsp$ and its optimal assignment $\bar{x}$ w.r.t. $w'$, $\cost(\bar{x}, \bar{y}) \leq \cst(\bar{y})$.
\end{lemma}
\begin{proof}
For every client $v\in P'$ let $v'$ denote the client guaranteed by Claim~\ref{clm:3gamma}; $\forall u'\in F'(v')$, $d(v, u') \leq 3\gamma_v$. Moreover, we construct an assignment of $\bar{y}$ denoted as $\hat{x}$ as follows. For each $v\in P'$, $\hat{x}_{vu} := \bar{x}_{vu}$ if $u\in G(v)$. Next, we consider the facilities in $F'(v')$ in an arbitrary order $u'_1, \cdots u'_\ell$ and process them in this order one by one. For each $j\leq \ell$, we set $\hat{x}_{vu'_j} := \min(y'_{u'_j}, (1- \sum_{u\in G(v)} \hat{x}_{vu}-\sum_{i<j}\hat{x}_{v u'_i}))$. Finally, for the remaining facilities $u'\in \sF \setminus (G(v)\cup F'(v'))$, we set $\hat{x}_{vu'}=0$.
Since for each client $v\in P'$, $1/2 \leq \sum_{u\in F'(v)} y'_u\leq \sum_{u\in G(v)} y'_u$, the constructed assignment $\hat{x}$ is a feasible assignment---i.e., $(\hat{x},\bar{y})$ is a feasible solution of $\cllp(w', \sM)$. Moreover, our constructions ensures that for every client $v\in P'$, $\sum_{u\in \sF}\hat{x}_{vu} = \sum_{u\in G(v)\cup F'(v')}\hat{x}_{vu} = 1$.
Finally, by the optimality of the assignment $\bar{x}$ w.r.t. $\bar{y}$ and $w'$, $\cost(\bar{x}, \bar{y}) \leq \cost(\hat{x},\bar{y})$. 

\begin{align*}
    \cost(\bar{x},\bar{y})
    &\leq \cost(\hat{x}, \bar{y})\\
    &= \sum_{u \in \sF} f(u) \cdot \bar{y}_u + \sum_{v\in P'} w'(v) \cdot \sum_{u\in \sF} d(v, u)^p \cdot \hat{x}_{vu} \\
    &= \sum_{u \in \sF} f(u) \cdot \bar{y}_u + \sum_{v\in P'} w'(v) \cdot \big(\sum_{u\in G(v)} d(v, u)^p \cdot \hat{x}_{vu} + \sum_{u\in F'(v')} d(v, u)^p \cdot \hat{x}_{vu}\big) \\
    &\leq \sum_{u \in \sF} f(u) \cdot \bar{y}_u + \sum_{v\in P'} w'(v) \cdot \big(\sum_{u\in G(v)} d(v, u)^p \cdot \hat{x}_{vu} + 3^p\gamma_v^p\sum_{u\in F'(v')} \hat{x}_{vu}\big) &&\rhd\text{by Claim~\ref{clm:3gamma}}\\
    &= \sum_{u \in \sF} f(u) \cdot \bar{y}_u + \sum_{v\in P'} w'(v) \cdot \big(\sum_{u\in G(v)} d(v, u)^p \cdot \hat{x}_{vu} + 3^p \gamma^p (1-\sum_{u\in G(v)}\hat{x}_{vu})\big)
    &&\rhd\sum_{u\in G(v)\cup F'(v')}\hat{x}_{vu} = 1\\
    &= \sum_{u \in \sF} f(u) \cdot \bar{y}_u + \sum_{v\in P'} w'(v) \cdot \big(\sum_{u\in G(v)} d(v, u)^p \cdot \bar{x}_{vu} + 3^p \gamma^p (1-\sum_{u\in G(v)}\bar{x}_{vu})\big) &&\rhd \forall u\in G(v), \hat{x}_{vu} = \bar{x}_{vu}\\
    &= \cst(\bar{y})
\end{align*}
\end{proof}
Next, we show that there exist a half-integral solution $y''$ that minimizes the modified cost function $\cst$ over the set of solutions described by $\solsp$.
\begin{lemma}\label{lem:half-integral-opt}
There is a half-integral solution $y''$ that minimizes $\cst$ over the polytope $\sP$. 
\end{lemma}
\begin{proof}
For the proof of this Lemma, we refer to \citep[Appendix A]{swamy2016improved} where 
it is shown that the polytope $\solsp$ has half-integral extreme solution. Hence, there is a polynomial time algorithm to find a half-integral solution $y''\in \solsp$ that minimizes the {\em linear} cost function $\cst$.

Note that by Claim~\ref{clm:feasible-polytope}, $(x'', y'')$ is a feasible solution of $\cllp(w', \sM)$ where $x''$ is the optimal (feasible) assignment w.r.t. $y''$ and $w'$.
\end{proof}
Finally, we have all the pieces to prove Theorem~\ref{thm:main-half-integral}.

\begin{proofof}{\bf Theorem~\ref{thm:main-half-integral}:}
By Lemma~\ref{lem:feasible-intermediate-sol}, $(x',y')$ is a feasible solution to $\cllp(w',\sM)$ and its cost is at most $\cst(y')\leq 3^p\cdot z^*$. Moreover, Lemma~\ref{lem:feasible-intermediate-sol} shows that the solution $y'$ is contained in the polytope $\solsp$. 
Then, by an application of Lemma~\ref{lem:half-integral-opt}, there exists a half-integral solution $y''$ such that $\cst(y'') \leq \cst(y')$---in fact, the solution $y''$ minimizes $\cst$ in the polytope $\solsp$. Now, we consider the half-integral solution $(x'', y'')$ of $\cllp(w', \sM)$ where $x''$ is the optimal assignment w.r.t. $y''$ and $w'$.
\begin{align*}
    \cost(x'', y'') 
    &\leq \cst(y'') &&\rhd\text{by Lemma~\ref{lem:proxy-cost-sol-space}} \\
    &\leq \cst(y') &&\rhd \text{by the optimality of $y''$ w.r.t. $\cst$ in the polytope $\solsp$} \\
    &\leq 3^p \cdot z^* &&\rhd \text{by Lemma~\ref{lem:feasible-intermediate-sol}}
\end{align*}
\end{proofof}

\subsection{Converting $(x'',y'')$ to an integer solution}
In this section, we show how to convert the half-integral solution $(x'', y'')$ of the well-separated instance to an integral solution of the well-separated instance without losing more than $e^{O(p)}\cdot z^*$ in the cost.

\begin{theorem}\label{thm:main-integral}
Let $z^*$ denote the cost of an optimal solution of $\cllp(w, \sM)$. There exists an integral solution of $\cllp(w', \sM)$ of cost at most $(4\cdot 3^{p-1} + 2) \cdot 3^p \cdot z^*$.
\end{theorem}

Similarly to the notion {\em fractional distance} $\sR$ defined w.r.t. the optimal solution of $(x,y)$ of the original instance (see Step I in Section~\ref{sec:half-integral}), for every client $v\in P'$, we define the fractional distance of $v$ w.r.t. $(x'', y'')$ as $\sR''(v):= \big(\sum_{u\in \sF} d(u,v)^p \cdot x''_{uv}\big)^{\frac{1}{p}}$. Moreover, for each client $v\in P'$, we denote the set of {\em serving} facilities of $v'$ in $(x'', y'')$ as $\flt(v) := \{u\in \sF: x''_{vu}>0\}$.

\paragraph{Step III: Identify {\core} clients.}\label{step:alg:core}
First, in Algorithm~\ref{alg:new-cluster-centers}, we construct a subset of clients $P'' \subseteq P'$, called {\em core clients} and a mapping $cr: P' \rightarrow P''$. 
The crucial property of the core clients $P''$ is the following: every facility $u \in \sF$ is serving {\em at most} one client in $P''$. In other words the family of sets $\{\flt(v)\}_{v\in P''}$ are disjoint.
\begin{algorithm}[h]
	\begin{algorithmic}[1]
		\STATE {\bfseries Input:} $(x'',y'')$: Half-integral solution of $\cllp(w',\sM)$ from Theorem \ref{thm:main-half-integral}
		\STATE $\sR''(v) \gets \left(\sum_{u\in \sF}  d(v,u)^p \cdot x''_{vu}\right)^{1/p}$ \textbf{for all} $v\in P'$	
        \STATE $\flt(v) \gets \{u:x''_{uv}>0\}$ \textbf{for all} clients $v\in P'$
        \STATE $P'' \gets \emptyset$
		\WHILE{$P' \neq \emptyset$}
		\STATE $v^*$ = $\argmin_{v\in P'} \sR''(v)$ \label{alg:line:minR}
        \STATE $P' \leftarrow P' \setminus \{v^*\}$, $P''\leftarrow P'' \cup \{v^*\}, cr(v^*) \leftarrow v^*$
        \FORALL{$v'\in P'$}
            \IF{$\flt(v^*)\cap \flt(v')\neq \emptyset$} \label{alg:line:disjoint-core}
                \STATE $P'\leftarrow P'\setminus \{v'\}, cr(v')= v^*$
            \ENDIF    
        \ENDFOR
		\ENDWHILE
        \RETURN $P''$, $cr$
    \end{algorithmic}
	\caption{constructs a set of core clients.}
	\label{alg:new-cluster-centers}
\end{algorithm}

\begin{claim}\label{clm:non-increasing-R}
For every client $v\in P'$, $\sR''(cr(v))\leq \sR''(v)$.
\end{claim}
\begin{proof}
The inequality holds trivially for the core clients $v\in P''$. Let $v'\in P'\setminus P''$. At the iteration in which $cr(v)$ is added to $P''$, $v$ is still present in $P'$. Hence, by the condition at line~\ref{alg:line:minR}, $\sR''(cr(v))\leq \sR''(v)$.
\end{proof}
\paragraph{Step IV: Obtaining an integral solution $(\Tilde{x},\Tilde{y})$.}
Similarly to our approach for constructing the half-integral solution $(x'', y'')$, we first construct an ``intermediate'' solution $\tilde{y}'$ with certain structures. Later, we exploit the known results in matroid intersection to find a ``good'' integral solution for the constructed intermediate solution.

$$
\tilde{y}'_u :=
\begin{cases}
      x''_{vu} &\text{if there exists a client }  v\in P'' \text{ such that } u\in \flt(v)\\
      y''_u &\text{otherwise}
\end{cases} 
$$

\begin{lemma}\label{lem:feasible-tilde-y}
The solution $(\tilde{x}', \tilde{y}')$ where $\tilde{x}'$ is an optimal assignment w.r.t. $\tilde{y}'$ and $w'$ is a feasible solution of $\cllp(w', \sM)$.
\end{lemma}
\begin{proof}
The constraints~\eqref{casgn} and~\eqref{under_assignment} are satisfied by the way $\tilde{x}'$ is constructed via Algorithm~\ref{alg:assign_optimal}. The matroid constraint, constraint~\eqref{matroid_constraint}, holds since for every facility $u\in \sF$, $\tilde{y}'_u\leq y''_u$ and $y''$ satisfies the matroid constraint.
\end{proof}
\begin{claim}
For every core client $v\in P''$, $\sum_{u\in \flt(v)}\tilde{y}'_u = 1$.
\end{claim}
\begin{proof}
\begin{align*}
    \sum_{u\in \flt(v)}\Tilde{y}'_u
    &= \sum_{u\in \flt(v)}x''_{vu} &&\rhd\text{by the definition of $\tilde{y}'$}\\
    &= \sum_{u\in \sF}x''_{vu} &&\rhd\flt(v):=\{u\in \sF:x''_{vu}>0\}\\
    &= 1 &&\rhd\text{by the feasibility of $(x'',y'')$ for $\cllp$}
\end{align*} 
\end{proof}
Next, for every client $v\in P'$, we define the {\em primary} $\pri_v$ and the {\em secondary} $\sec_v$ facilities of $v$ such that $\pri_v$ denotes the nearest facilities to $v$ among the possibly two facilities serving $v$ in the half-integral solution $(x'',y'')$. Note that since $(x'', y'')$ is a half-integral solution, either $x''_{v\pri_v} = x''_{v\sec_v} = \frac{1}{2}$ or $x''_{v\pri_v} = 1$ otherwise. For technical reason, in the latter case, we set $\sec_v = \pri_v$.

\begin{claim}\label{claim:R-half-integr-def}
For every client $v\in P'$, $\sR''(v)^p=\frac{1}{2}(d(\pri_v,v)^p+d(\sec_v,v)^p)$.
\end{claim}
\begin{proof}
It simply follows from the half-integrality of the solution $(x'', y'')$ and the definition of the primary and the secondary facilities. 
%
\end{proof}

\begin{claim}\label{claim:DistCmpR}
For every client $v\in P'$, $d(\pri_v,v)^p\leq \sR''(v)^p\leq d(\sec_v,v)^p\leq 2\sR''(v)^p$.
\end{claim}
\begin{proof}
By the definition of $\pri_v$ and $\sec_v$, $d(\pri_v,v)^p \leq \frac{1}{2}(d(\pri_v,v)^p+d(\sec_v,v)^p) \leq d(\sec_v,v)^p$. By an application of Claim~\ref{claim:R-half-integr-def} the proof is complete.
\end{proof}

Similarly to the previous section, we introduce a cost function $\pcst$ that serves as a proxy to bound the cost of a set of half-integral solutions we are considering in this section.
Let $\pcst(\tilde{y}'):= \sum_{u\in\sF}f(u)\cdot \tilde{y}'_u + \sum_{v\in P'}w'(v) A_v(\tilde{y}')$ where $A_v(\tilde{y}')$ is the proxy for the per-unit assignment cost of a client $v\in P'$ which is defined as
\begin{align}\label{eq:assignment-cost-int}
A_v(\tilde{y}'):= \sum_{u\in \flt(cr(v))} d(u,v)^p \cdot \tilde{y}'_u
\end{align}
\begin{lemma}\label{lem:cost-compare-half-integral}
The proxy cost of the intermediate solution $\tilde{y}'$ is at most $(4\cdot 3^{p-1} + 2)$ times the cost of $(x'',y'')$; $\pcst(\tilde{y}') \leq (4\cdot 3^{p-1} + 2) \cdot \cost(x'', y'')$.
\end{lemma}
\begin{proof}
Since for every facility $u\in \sF$, $\tilde{y}'_u \leq y''_u$, $\sum_{u\in \sF} f(u)\cdot \tilde{y}'_u \leq \sum_{u\in \sF} f(u)\cdot y''_u$. 
Next, we consider the following cases to bound the contribution of the assignment cost of a client $v$ in $\pcst(\Tilde{y}')$.
\begin{enumerate}[leftmargin=*]
    \item{\bf $v$ is a core client ($v\in P''$)}
        \begin{align*}
            A_v(\Tilde{y}')
            &= \sum_{u\in \flt(cr(v))} d(v, u)^p \cdot \tilde{y}'_u \\ 
            &= \sum_{u\in \flt(v)} d(v, u)^p \cdot \tilde{y}'_u &&\rhd cr(v) = v\\
            &= \sum_{u\in \flt(v)} d(v, u)^p \cdot x''_{vu} &&\rhd\text{by the definition of $\tilde{y}'$}\\
            &= \sum_{u\in \sF} d(v, u)^p \cdot x''_{vu} &&\rhd\flt(v):=\{u\in \sF:x''_{vu}>0\} \\
            &= \sR''(v)^p
        \end{align*} 
    \item{\bf $v$ is not a core client ($v\in P\setminus P''$) and $\pri_v \in \flt(cr(v))$.} Let $u^*\in \flt(cr(v))\setminus \{\pri_v\}$.
        \begin{align*}
            A_v(\tilde{y}')
            &= \sum_{u\in \flt(cr(v))} d(v, u)^p \cdot \tilde{y}'_{u} \\ 
            &= d(v,\pri_v)^p \cdot \tilde{y}'_{\pri_v} + d(v,u^*)^p \cdot \tilde{y}'_{u^*} &&\rhd\flt(cr(v))=\{u^*,\pri_v\} \\
            &\leq d(v, \pri_v)^p + 3^{p-1}\cdot \big(d(v, \pri_v)^p + d(\pri_v, cr(v))^p + d(cr(v), u^*)^p\big) &&\rhd\text{Eq.~\eqref{ineq:2hop-triangle-ineq} and $\left\|\tilde{y}'\right\|_{\infty}\leq 1$}\\
            &\leq \sR''(v)^p + 3^{p-1}\cdot \big(\sR''(v)^p+ 2\sR''(cr(v))^p\big) &&\rhd\text{Claim~\ref{claim:R-half-integr-def} and~\ref{claim:DistCmpR}}\\
            &\leq \sR''(v)^p + 3^p \cdot \sR''(v)^p &&\rhd\text{Claim~\ref{clm:non-increasing-R}}\\
            &\leq (3^p +1) \cdot \sR''(v)^p 
        \end{align*}     
    \item{\bf $v$ is not a core client ($v\in P\setminus P''$) and $\pri_v \notin \flt(cr(v))$.} Since $\pri_v\notin \flt(cr(v))$, we have that $\sec_v\in \pri_v\notin \flt(cr(v))$. Let $u^*\in \flt(cr(v))\setminus \{\sec_v\}$. 
        \begin{align*}
            A_v(\Tilde{y}')
            &= \sum_{u\in \flt(cr(v))} d(v,u)^p \cdot \tilde{y}'_u\\ 
            &= d(v, \sec_v)^p \cdot \tilde{y}'_{\sec_v} + d(v, u^*)^p \cdot \tilde{y}'_{u^*} &&\rhd\flt(cr(v))=\{u^*,\sec_v\} \\
            &\leq 2\sR''(v)^p + d(v, u^*)^p &&\rhd\text{Claim~\ref{claim:DistCmpR}, $\Tilde{y}'_{\sec_v}\leq 1$}\\
            &\leq 2\sR''(v)^p + 3^{p-1}\cdot\big(d(v, \sec_v)^p + d(\sec_v,cr(v))^p + d(cr(v), u^*)^p\big) &&\rhd\text{Eq.~\eqref{ineq:2hop-triangle-ineq}}\\
            &\leq 2\sR''(v)^p + 3^{p-1}\cdot (2\sR''(v)^p+2\sR''(cr(v)))^p &&\rhd\text{Claim~\ref{claim:R-half-integr-def} and~\ref{claim:DistCmpR}}\\
            &\leq 2\sR''(v)^p + 4\cdot 3^{p-1} \cdot \sR''(v)^p &&\rhd\text{Corollary~\ref{clm:non-increasing-R}}\\
            &\leq (4\cdot 3^{p-1} + 2) \cdot \sR''(v)^p
        \end{align*}    
\end{enumerate}
Hence, summing over all clients in $P'$,
\begin{align*}
    \pcst(\tilde{y}') 
    &= \sum_{u\in\sF}f(u)\cdot \tilde{y}'_u + \sum_{v\in P'} w'(v) \cdot A_v(\tilde{y}') \\
    &\leq \sum_{u\in\sF}f(u)\cdot \tilde{y}'_u + (4\cdot 3^{p-1} + 2)\cdot \sum_{v\in P'} w'(v)\cdot \sR''(v)^p \\
    &\leq \sum_{u\in\sF}f(u)\cdot y''_u + (4\cdot 3^{p-1} + 2)\cdot \sum_{v\in P'} w'(v) \cdot \big(\sum_{u\in \sF} d(u,v)^p \cdot x''_{uv}\big) \\
    &\leq (4\cdot 3^{p-1} + 2)\cdot \cost(x'', y'')
\end{align*}

\end{proof}
Next, we define the following polytope $\intSolSp$ that has integral extreme points and contains $\tilde{y}'$. 
\begin{align}\label{eq:int-sol-space}
\intSolSp := \{y \in \mathbb{R}^\mathcal{F}_+:\sum_{u\in S}y_u \leq r(S) \quad \forall S\subseteq \sF,\quad \sum_{u\in \flt(v)} y_u=1 \quad \forall v \in P''\}
\end{align}
First we show that for every solution $\bar{y}\in \intSolSp$, $\pcst(\bar{y}) \geq \cost(\bar{x}, \bar{y})$ where $\bar{x}$ is an optimal assignment w.r.t.~$\bar{y}$ and $w'$ as described in Algorithm~\ref{alg:assign_optimal}.
We now prove that the cost of any vector $\Tilde{y}\in \intSolSp$ with its optimal assignment $\Tilde{x}$ obtained by Algorithm \ref{alg:assign_optimal} is at most $H(\Tilde{y})$. This Lemma proves for both $\Tilde{y}'$ and for $\Tilde{y}$ that $H(\Tilde{y}')$ and respectively $H(\Tilde{y})$ is an upper bound on the assignment cost.
\begin{lemma}\label{lem:pcst-upper-bound}
For every $\bar{y}\in \intSolSp$, $\cost(\bar{x},\bar{y})\leq H(\bar{y})$ where $\bar{x}$ is an optimal assignment w.r.t.~$\bar{y}$ and $w'$.
\end{lemma}
\begin{proof}
The total contribution of the facility opening cost in $\cost(\bar{x},\bar{y})$ and $\pcst(\bar{y})$ are the same.
Observe that for every client $v'\in P'$, there exists a core client $v \in P''$ such that $cr(v')=v$. 
In the following, we construct a feasible assignment $\hat{x}$ w.r.t.~$\bar{y}$ and $P''$ such that its assignment cost is equal to the assignment cost of $\pcst(\tilde{y}')$ (which is equal to $\sum_{v\in P'} A_v(\bar{y})$). Once we have $\hat{x}$, the lemma simply follows from the optimality of assignment $\bar{x}$ w.r.t.~$\bar{y}$ and $P''$.

We construct $\hat{x}$ as follows. For each client $v\in P'$, $\hat{x}_{vu} = \bar{y}_u$ if $u\in \flt(cr(v))$ and zero otherwise. This is a feasible assignment w.r.t. $\bar{y}$ and $w'$ because $\bar{y}\in \sQ$. We next bound the cost of solution $(\hat{x},\bar{y})$.
\begin{align*}
    \cost(\bar{x}, \bar{y}) 
    &\leq \cost(\hat{x}, \bar{y}) &&\rhd\text{by the optimality of $\bar{x}$} \\
    &= \sum_{u\in \sF} f(u)\cdot \bar{y}_u + \sum_{v\in P'} \sum_{u\in \sF} w'\cdot d(v,u)^p\cdot \hat{x}_{vu} \\
    &=\sum_{u\in \sF} f(u)\cdot \bar{y}_u + \sum_{v\in P'} \sum_{u\in \flt(cr(v))} w'(v)\cdot d(v,u)^p\cdot \bar{y}_u &&\rhd\text{by the definition of $\tilde{x}$}\\
    &=\pcst(\bar{y})
\end{align*}
\end{proof}

\begin{lemma}\label{lem:integral-opt}
There is an integral solution $\tilde{y}$ that minimizes $\pcst$ over the polytope $\intSolSp$. 
\end{lemma}
\begin{proof}
The desired solution $\tilde{y}$ exists since $\pcst$ is a linear function and the extreme point of the polytope $\intSolSp$ are integral. The latter holds since $\intSolSp$ is non-empty and an intersection of two matroid polytopes (defined by $\sM$ and the partition matroid corresponding to $\sum_{u\in \flt(v)}y_u=1, \forall v\in P''$).  
\end{proof}

\begin{proofof}{\bf Theorem~\ref{thm:main-integral}}
By Theorem~\ref{thm:main-half-integral}, there exist a half-integral solution $(x'', y'')$ of $\cllp(\sM, w')$ of cost at most $3^p \cdot z^*$.
Let $\tilde{y}$ be the minimizer of $\pcst$ over the polytope $\intSolSp$. Moreover, let $\tilde{x}$ denote the optimal assignment of $\tilde{y}$ w.r.t.~$w'$. By Lemma~\ref{lem:integral-opt}, the solution $(\tilde{x}, \tilde{y})$ is an integral feasible solution of $\cllp(\sM, w')$. Furthermore,
\begin{align*}
    \cost(\tilde{x}, \tilde{y}) 
    &\leq \pcst(\tilde{y}) &&\rhd\text{by Lemma~\ref{lem:pcst-upper-bound}} \\
    &\leq \pcst(\tilde{y}') &&\rhd\text{since $\tilde{y} = \argmin_{y\in \intSolSp}\pcst(y)$ and $\tilde{y}'\in \intSolSp$}\\
    &\leq (4\cdot 3^{p-1} + 2) \cdot \cost(x'', y'') &&\rhd\text{by Lemma~\ref{lem:cost-compare-half-integral}}\\
    &\leq (4\cdot 3^{p-1} + 2) \cdot 3^p \cdot z^* &&\rhd\text{by Theorem~\ref{thm:main-half-integral}}
\end{align*}
\end{proofof}

Now we are ready to state the main theorem of $\ell_p$-norm facility location under matroid constraint.
\begin{theorem}[Main Theorem of $\ell_p$-norm Facility Location Under Matroid Constraint]\label{thm:main-facility}
For $p>1$, there exists a polynomial time algorithm that finds a $(16^p)$-approximate solution of $\ell_p$-clustering on $(w, P)$ under matroid constraint $\sM$.
\end{theorem}
\begin{proof}
Following the result of this section, we first construct a well-separated instance $(w', P')$. By Theorem~\ref{thm:main-integral}, we can construct an integral solution of the well-separated instance of cost at most $(4\cdot 3^{p-1} + 2) \cdot 3^p \cdot z^*$. Next, by Theorem~\ref{thm:convert-to-original}, the constructed solution can be extended to a feasible solution of the original instance $(w, P)$ of cost at most $4 \cdot 16^{p-1}\cdot z^* + (\frac{8}{7})^{p-1}\cdot (4\cdot 3^{p-1} + 2) \cdot 3^p \cdot z^* < 16^p \cdot z^*$ (for $p>1$).  
\end{proof}
\begin{remark}
Note that our approach works for $p=1$ too and achieves a $22$-approximation guarantee. However, since the result of~\citep{swamy2016improved} provides an $8$-approximation in this case ($p=1$), we only consider $p>1$ here.
\end{remark}

\section{Missing Proofs}\label{sec:missing-proofs}
\begin{lemma}[Lemma~A.1~\cite{makarychev2019performance}]\label{lem:p-norm-ineq}
Let $x, y_1, \cdots, y_n$ be non-negative real numbers and $\lambda>0, p\ge 1$. Then,
\begin{align*}
    (x + \sum_{i=1}^n y_i)^p \leq (1+\lambda)^{p-1} x^p + \Big(\frac{(1+\lambda)n}{\lambda}\Big)^{p-1} \sum_{i=1}^n y_i^p.
\end{align*}
\end{lemma}

\begin{proofof}{\bf Lemma~\ref{lem:d'-metric}.}
Let $u,v,w \in (\sF\cup \sM)$ three arbitrary points and let $u_P,v_P,w_P$ be their corresponding points from $P$. Furthermore, let $\hat{\varepsilon}:=\min\{(\frac{\eps(n-k)}{\beta \cdot k})^{1/p},1\}$.

First we prove that $d'(u,v)=0 \iff u=v$.
If $u=v$, then by line \ref{alg:reduction:zero} the distance $d'(u,v)$ is set to zero. To show the other direction, if $d'(u,v)=0$ then the constraint $u=v$ for the assignment in line \ref{alg:reduction:zero} is satisfied since $d(u_p,v_p)>0$ for all $u_p\neq v_p$ (line \ref{alg:reduction:same}) and $d'(u,v) = \hat{\varepsilon}\cdot \delta>0$ when $u_p = v_p$ and $u\neq v$ (line \ref{alg:reduction:epsilon}).

Secondly, we prove the symmetric property $d'(u,v)=d'(v,u)$.
If $d'(u,v)=0$, then by the first part $u=v$ and therefore $d'(v,u)=0=d(u,v)$.
Assume $d'(u,v)> 0$ which implies $u \neq v$. If $u_P\neq v_P$, then by line \ref{alg:reduction:same} and the metric properties of $d$, $d'(u,v)=d(u_P,v_P)=d(v_P,u_P)=d'(v,u)$ holds. 
Otherwise, by line \ref{alg:reduction:epsilon},  $d'(u,v)=\hat{\varepsilon}\cdot \delta=d'(v,u)$.

Lastly we show that the triangle inequality $d'(u,w)\leq d'(u,v)+d'(v,w)$ holds. If $u=w$ then by the first property, $d'(u,w)=0$ so the inequality holds. 
Assume $u\neq w$ and consider their corresponding points $u_P, w_P$. 
\begin{enumerate}[leftmargin=*]
    \item If $u_P = w_P$ then, $d'(u,w)=\hat{\varepsilon}\cdot \delta$. Let $v_P$ be the corresponding point of $v$. If $v_P=u_P$, then $d'(u,v)=d'(u,w)=\hat{\varepsilon}\cdot \delta$ and therefore $d'(u,w)\leq d'(u,v)+d'(v,w)$ already holds. If $v_P\neq u_P$, then $d'(u,v)=d(u_P,v_P)\geq \min_{x,y\in P}d(x,y)\geq \hat{\varepsilon}\cdot \delta$. Thus $d'(u,w)\leq d'(u,v)+d'(v,w)$ holds.
    \item If $u_P \neq w_P$ then $d'(u,w)=d(u_P,w_P)\geq \hat{\varepsilon}\cdot \delta$. Note that $(u_P=v_P \text{ and } v_P=w_P)$ can not hold, so consider the remaining three cases cases:
    \begin{enumerate}[leftmargin=*]
        \item $v_P = w_P$ and $u_P\neq v_P$. Then $d'(u,w) = d'(u,v)$ and therefore $d'(u,w) \leq d'(u,v)+d'(v,w)$
        \item $u_P = v_P$ and $v_P\neq w_P$ . Then $d'(u,w) = d'(v,w)$ and therefore $d'(u,w) \leq d'(u,v)+d'(v,w)$
        \item $u_P\neq v_P$ and $v_P\neq w_P$. Then $d'(u,w)=d(u_P,w_P), d'(u,v)=d(u_P,v_P), d'(v,w)=d(v_P,w_P)$ and since $d(\cdot)$ satisfies the triangle inequality, $d'(u,w)\leq d'(u,v)+ d'(v,w)$ holds.
    \end{enumerate}
\end{enumerate}
\end{proofof}

\begin{proofof}{\bf Theorem~\ref{thm:k-center-reduction}}
Let $\cllc$ be a $\beta$-approximation algorithm for $k$-center under partition matroid constraint. Consider an instance of $\alpha$-fair $k$-center on $P$ and let $(P', \{(\bar{P}_0,k_0), (\bar{B}_1, k_1), \cdots, (\bar{B}_m, k_m)\})$ be the instance of $k$-center under partition matroid constraint constructed by Algorithm~\ref{alg:k-center-reduction} with input parameters $P, k$ and $\alpha$. We show that the solution returned by $\cllc(P', \{(\bar{P}_0,k_0), (\bar{B}_1, k_1), \cdots, (\bar{B}_m, k_m)\})$ can be converted to a $(\beta+\eps,3)$-bicriteria approximate solution of the given instance of $\alpha$-fair $k$-center on $P$.

Let $\sB= \{B_1, \cdots, B_m\}$ be the critical regions of $P$ constructed in Algorithm~\ref{alg:AlphaCluster}.
Let $\sol_C$ be the solution returned by $\cllc(P', \{(\bar{P}_0,k_0), (\bar{B}_1, k_1), \cdots, (\bar{B}_m, k_m)\})$ and let $\opt$ be an optimal solution of $\alpha$-fair $k$-center of $P$. Note that since adding centers in $\sol_C$ only reduces the $k$-center cost of the solution on $(P', \{(\bar{P}_0,k_0), (\bar{B}_1, k_1), \cdots, (\bar{B}_m, k_m)\})$, without loss of generality we can assume that $\sol_C$ picks exactly one center from each of $B_i$, for $i\in [m]$, and exactly $k-m$ centers from $\bar{P}_0$.
Now we construct a solution $\sol$ of $\alpha$-fair $k$-center on $P$ using the solution $\sol_C$. We start with an initially empty set of centers $\sol$. In the first step, for each $B\in \sB$, let $c_i$ denote the center in $\sol_C \cap \bar{B}_i$ and then we add the point $c\in P$ corresponding to $c_i$ to $\sol$. 
Next, in the second step, for each $o_0\in  \sol_C \cap \bar{P}_0$, we add the point $o\in  P$ corresponding to $o_0$ to $\sol$. Note that as some of these points may have already been added to $\sol$ in the first step, the final solution has at most $k$ distinct centers.

\paragraph{Fairness approximation.} 
By the first step in the construction of $\sol$, for each $i\in [m]$, $|B_i \cap \sol|\geq 1$. Hence, by Lemma~\ref{lem:fair-sol-critical}, $\sol$ is a $(3\alpha)$-fair $k$-center clustering of $P$.
    
\paragraph{Cost approximation.} First we show that the cost of $\sol_C$ on $P'$ is not smaller than the $k$-center clustering cost of $P$ using $\sol$. Lets assume that there exist $v\in P$ such that $d(v, \sol) > d'(v',\sol_C)$ where $v'$ is a copy of $v$ in $P'$. Let $c'$ be the closest center to $v'$ in $\sol_C$. Let $c$ denote the point in $P$ corresponding to $c'$. Since after the second step of constructing $\sol$ all original copies of the centers in $\sol_C$ are added to $\sol$, $c\in \sol$. Hence, $d(v,\sol) \leq d(v, c) \leq d'(v',c') = d'(v', \sol_C)$ which is a contradiction. Hence, the cost of $\sol_C$ is not smaller than the cost of $\sol$. 
    
Next, we bound the cost of $\sol_C$ on $P'$ in terms of the cost of $k$-center clustering of $P$ using $\opt$. By the definition of $\alpha$-fairness, each point $v\in P$ must have a center in $\opt$ within distance at most $\alpha \cdot r(v)$. Hence, for each critical region $B\in \sB$, $|\opt\cap B|\geq 1$. For each $i\in [m]$, let $c_i$ be the copy of an arbitrary center $c\in \opt\cap B_i$ in the set $\bar{B}_i$. For the remaining points in $\opt$, we pick their corresponding copies in the set $\bar{P}_0$. Let $\opt_{C}$ denote the constructed solution for the instance $P'$. 
Since $\opt_C$ picks exactly one point from each set $\bar{B}_i$, for $i\in [m]$, and exactly $k-m$ points from $\bar{P}_0$, $\opt_C$ is a feasible solution for $k$-center under partition matroid constraint on instance $(P', \{(\bar{P}_0,k_0), (\bar{B}_1, k_1), \cdots, (\bar{B}_m, k_m)\}, d')$. Moreover, since for every pair $(v\in P, c\in \opt)$, there exists a pair $(v'\in P', c'\in \opt_C)$ such that $d'(v', c') \leq d(v, c) + \eps \cdot \delta/\beta$, $\cost_{\kcenter}(\opt_C; P') \leq \cost_{\kcenter}(\opt; P) + \frac{\eps\cdot \delta}{\beta}$. Hence,
\begin{align*}
    \cost_{\kcenter}(\sol, P) 
    \leq \cost_{\kcenter}(\sol_C, P') \leq \beta\cdot \cost_{\kcenter}(\opt_C, P') 
    &\leq \beta \cdot (\cost_{\kcenter}(\opt, P) + \frac{\eps\cdot \delta}{\beta})\\
    &\leq (\beta+\eps) \cdot \cost_{\kcenter}(\opt, P),
\end{align*}
where the last inequality follows since $\cost_{\kcenter}(\opt, P) \geq \delta$.
Thus, the $k$-center clustering cost of $P$ using $\sol$ is within a $\beta+\eps$ factor of the cost of any optimal $\alpha$-fair $k$-center of $P$.  
\end{proofof}

\end{document}